\documentclass[sigconf]{acmart}
\usepackage[ruled,vlined]{algorithm2e}
\usepackage{multirow}
\usepackage{graphicx}
\usepackage{subcaption}
\usepackage{booktabs}
\usepackage[most]{tcolorbox}
\usepackage{balance}

\usepackage{listings}
\AtBeginDocument{%
  }

\copyrightyear{2026}
\acmYear{2026}
\setcopyright{cc}
\setcctype{by}
\acmConference[KDD 2026]{Proceedings of the 32nd ACM SIGKDD Conference on Knowledge Discovery and Data Mining V.2}{August 9--13, 2026}{Jeju Island, Republic of Korea}
\acmBooktitle{Proceedings of the 32nd ACM SIGKDD Conference on Knowledge Discovery and Data Mining V.2 (KDD 2026), August 9--13, 2026, Jeju Island, Republic of Korea}
\acmISBN{979-8-4007-2259-2/2026/08}
\acmDOI{10.1145/3770855.3818134}
\acmSubmissionID{v2rtp3489}

\begin{document}
\title{State Machine Guided Multi-Relational Synthetic Data from Logs for Anomaly Detection (Extended Version)}


\author{Aja Khanal}
\affiliation{%
  \institution{University of Western Ontario}
  \city{London}
  \country{Canada}}
\email{akhanal3@uwo.ca}

\author{Apurva Narayan}
\affiliation{%
  \institution{University of Western Ontario}
  \city{London}
  \country{Canada}}
\email{apurva.narayan@uwo.ca}







\renewcommand{\shortauthors}{Khanal and Narayan}

\begin{abstract}
Software systems generate massive unstructured logs that record execution behavior, failures, and interactions across components, yet existing log anomaly detection methods treat these logs primarily as flat sequences of templates, overlooking the relational execution structure that governs how events co-occur and evolve over time. We propose a framework that discovers this hidden structure by recovering an execution state machine directly from logs and inducing a corresponding multi-table relational schema connecting traces, events, states, transitions, and parameters. This discovered state machine serves as a generative prior to produce realistic multi-relational synthetic data that preserves structural, temporal, and process constraints while amplifying rare but valid execution behaviors. We assess the fidelity of the generated data through constraint validation, distributional similarity, and process-level metrics, and demonstrate its usefulness by showing that augmenting real logs with the synthetic relational data significantly improves anomaly and bug detection on held-out real datasets compared to sequence-based baselines and naive oversampling. Our results show that execution logs implicitly encode a relational database governed by a latent state machine, and that recovering this structure enables principled synthetic data generation for robust and interpretable anomaly detection. 
\newline
Code: https://github.com/Idsl-group/LogSynthFSM-KDD-2026
\end{abstract}





\maketitle

\section{Introduction}

Modern software systems generate massive volumes of unstructured execution logs that record system behavior, failures, and interactions across components. These logs are central to anomaly detection and failure diagnosis, and have motivated extensive research in log mining, parsing, and log-based anomaly detection~\cite{he2017deeplog,du2017spell,zhang2019loganomaly,du2019loganomaly,fu2022survey}. However, most existing approaches model logs as flat sequences of events or templates, overlooking the relational and process-level structure that governs how executions evolve over time. This simplification limits both detection accuracy and interpretability, especially in large-scale, distributed systems where failures often arise from violations of execution flow rather than isolated events.

In practice, execution logs implicitly encode latent control flow, dependencies, and state transitions. Software executions proceed through recurring states, follow structured transition patterns, and involve relational interactions among parameters, resources, and components. Sequence-based models, including recurrent and Transformer-based approaches~\cite{he2017deeplog,zhang2019loganomaly,guo2021logbert,liao2025log}, can learn rich temporal dependencies and even FSA-like state-tracking behavior in controlled settings~\cite{zhang2025fsa}, but they do not by themselves explicitly enforce execution-state structure, transition legality, or relational consistency during generation. This distinction is important because unconstrained LLM-based extraction and generation pipelines can exhibit structural failures~\cite{tan2026structural}; in our ablation, removing the State Discovery Agent reduces transition validity from 0.985 to 0.714 and rare-event PR-AUC from 0.511 to 0.341, indicating that the gain comes from execution-aware constraints rather than generic sequential modeling alone.

Synthetic data generation has been explored to mitigate data sparsity and class imbalance in log-based anomaly detection~\cite{xu2018logbased,zhang2019logpai}. Existing log synthesis and augmentation methods typically operate at the template or sequence level and fail to respect underlying execution semantics, often producing unrealistic or structurally invalid traces~\cite{du2019loganomaly}. More fundamentally, these approaches assume that the structure governing log generation is known in advance or fixed, whereas in real systems this structure is latent, evolving, and must be inferred directly from raw logs~\cite{he2020loghub}.

Recent advances in large language models (LLMs) have demonstrated strong capabilities in understanding, summarizing, and restructuring unstructured text~\cite{brown2020language,wei2022chain}, suggesting their potential for extracting latent execution knowledge from logs. However, directly applying LLMs in a single-pass manner is insufficient due to context window limitations, lack of structural guarantees, and susceptibility to hallucination when generating long or complex execution traces. These challenges motivate the need for structured, coordinated reasoning and explicit constraint enforcement rather than isolated generation.

We propose \textbf{LogSynthFSM}, a multi-agent LLM framework that discovers latent execution structure from logs and leverages this structure to generate \textbf{multi-table relational synthetic data} for anomaly detection. LogSynthFSM decomposes the task into specialized agents that collaboratively parse logs, infer an execution state machine, induce a relational schema, and synthesize data under structural, temporal, and process constraints. The recovered state machine serves as a generative prior, enabling realistic multi-relational data generation while selectively amplifying rare but valid execution behaviors. We evaluate the generated data using constraint validity, distributional similarity, and process fidelity metrics, and demonstrate that augmenting real logs with synthetic relational data from LogSynthFSM significantly improves anomaly and bug detection on held-out real datasets compared to sequence-based baselines and naive oversampling. These results show that execution logs implicitly encode a relational database governed by a latent state machine, and that LLM-guided multi-agent structure discovery enables principled synthetic data generation for robust and interpretable anomaly detection.


\section{Background and Problem Formulation}

\subsection{Classical Methods}

Traditional log-based anomaly detection methods address system monitoring by modeling logs as sequences of events or templates. Early approaches rely on rule-based parsing and statistical modeling, such as invariant mining and frequency analysis~\cite{xu2009detecting,lin2016logclustering}, while more recent methods employ deep learning techniques including recurrent neural networks and Transformers to capture sequential patterns in logs~\cite{he2017deeplog,zhang2019loganomaly,du2019loganomaly,guo2021logbert}. These methods typically operate on template sequences extracted from raw logs using parsers such as Drain or Spell~\cite{he2017drain,du2017spell}, and detect anomalies as deviations from learned sequence distributions.

To mitigate data sparsity and class imbalance, several approaches explore log data augmentation and synthesis~\cite{xu2018logbased,zhang2019logpai}. Common strategies include oversampling rare templates, perturbing existing sequences, or interpolating hidden representations to generate additional training examples~\cite{chandola2009anomaly}. Despite improving coverage, these methods largely ignore the underlying execution semantics of software systems. In particular, they do not explicitly model execution states, transitions, or relational dependencies among events, parameters, and resources, often resulting in unrealistic or semantically incoherent synthetic traces that limit generalization to complex failure modes.

\subsection{LLM-Based Multi-Agent Systems}

Recent advances in large language models (LLMs) have enabled strong performance in understanding, reasoning over, and generating structured outputs from unstructured text~\cite{brown2020language,wei2022chain}. Beyond single-pass prompting, LLM-powered multi-agent systems have emerged as a promising paradigm for complex tasks that require coordination, decomposition, and iterative refinement. General frameworks such as Self-Instruct~\cite{wang2022self}, AutoGPT-style agentic workflows~\cite{significant2023autogpt}, and collaborative simulation systems~\cite{park2023generative} demonstrate that decomposing objectives across specialized agents can yield more robust and controllable behavior than monolithic generation.

In data synthesis and structured generation contexts, multi-agent LLM frameworks have been explored for graph, knowledge base, and relational data generation, leveraging semantic understanding while enforcing explicit structural constraints~\cite{du2024graphmaster,chen2023agent, khanal2026agentsofdiffusion}. These approaches show that agent specialization and feedback-driven refinement are effective for maintaining global consistency in complex outputs. However, existing methods typically assume that the underlying structure, such as a graph schema, ontology, or topology, is given a priori. In contrast, execution logs present a fundamentally different challenge: the governing structure is latent, noisy, and must be discovered directly from unstructured text. The problem of jointly discovering latent execution structure from logs and using it to guide multi-table relational synthetic data generation remains largely unexplored.
.

\subsection{Problem Formulation: Relational Synthesis}

\paragraph{Execution Logs.}
We consider a collection of execution logs represented as a set of traces
\[
\mathcal{L} = \{T_1, T_2, \dots, T_M\},
\]
where each trace $T_m = (e_1, e_2, \dots, e_{n_m})$ is an ordered sequence of events. Each event $e_i$ is associated with a timestamp, a log template, and a set of extracted parameters.

\paragraph{Execution State Machine.}
We assume that the logs are generated by an underlying execution process governed by a latent finite state machine (FSM)
\[
\mathcal{F} = (\mathcal{S}, \mathcal{A}, \delta),
\]
where $\mathcal{S} = \{s_1, s_2, \dots, s_K\}$ is a set of execution states, $\mathcal{A}$ denotes observable events, and $\delta : \mathcal{S} \times \mathcal{A} \rightarrow \mathcal{S}$ is the state transition function. The FSM is not observed directly and must be inferred from raw logs.

\paragraph{Knowledge Extraction.}
Given the length and heterogeneity of logs, we define a knowledge extraction function
\[
\Phi : \mathcal{L} \rightarrow \mathcal{K},
\]
which maps the log corpus to a compact knowledge representation
\[
\mathcal{K} = \Phi(\mathcal{L}) = (\mathcal{S}_k, \mathcal{T}_k, \Theta_k),
\]
where $\mathcal{S}_k$ is a set of inferred execution states, $\mathcal{T}_k \subseteq \mathcal{S}_k \times \mathcal{S}_k$ is a set of state transitions, and $\Theta_k$ represents state-conditional parameter and timing distributions. The extraction function $\Phi$ is implemented via a coordinated multi-agent LLM framework to ensure semantic coherence while remaining within context constraints.

\paragraph{Relational Schema Induction.}
Based on the extracted knowledge $\mathcal{K}$, we induce a multi-table relational schema
\[
\mathcal{R} = \{R_{\text{trace}}, R_{\text{event}}, R_{\text{state}}, R_{\text{transition}}, R_{\text{param}}\},
\]
which captures execution traces, events, states, transitions, and event parameters as interconnected relations with primary and foreign key constraints.

\paragraph{Synthetic Data Generation.}
We formalize multi-relational synthetic log generation as a function
\[
\Psi : \mathcal{K} \rightarrow \mathcal{D}_s,
\]
where $\mathcal{D}_s$ denotes a set of synthetic relational tables conforming to schema $\mathcal{R}$. The generation process samples valid state paths from the inferred FSM and synthesizes event and parameter records conditioned on state-specific distributions while enforcing relational, temporal, and process constraints.

\section{Methodology}

\begin{figure*}[t]
    \centering
    \includegraphics[width=\textwidth]{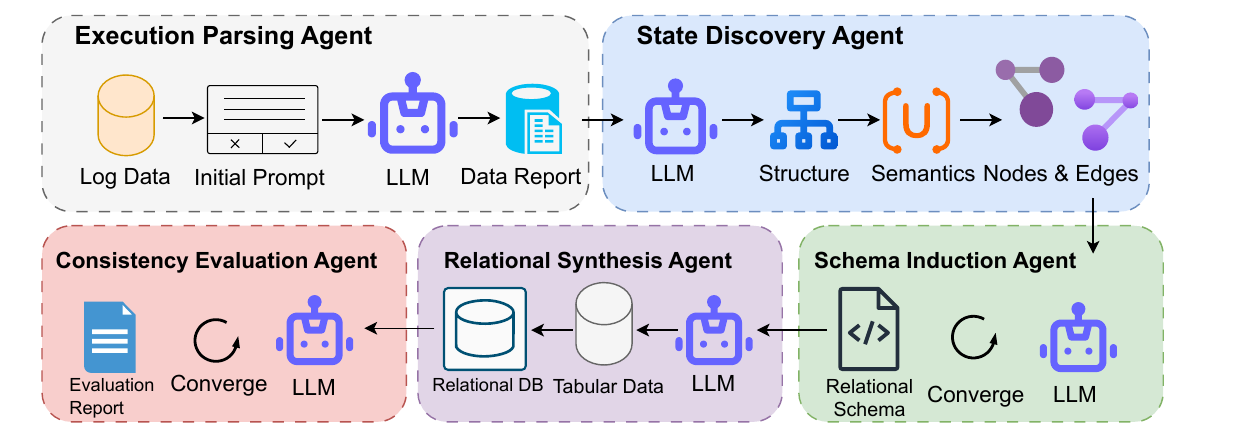}
    \caption{Overview of LogSynthFSM. The system converts raw logs into an execution-aware multi-relational representation, synthesizes realistic traces guided by an inferred state machine, and improves downstream anomaly detection through structure-preserving data generation.}
    \label{fig:overview}
\end{figure*}

Our framework, \textbf{LogSynthFSM}, adopts a coordinated multi-agent design in which specialized agents collaboratively discover execution structure from logs and generate multi-table relational synthetic data under structural and process constraints. Each agent is responsible for a distinct stage of the pipeline as shown in Figure \ref{fig:overview}.

\subsection{System Overview}

LogSynthFSM formulates multi-table relational synthetic data generation from execution logs as an iterative, structure-guided reasoning process implemented through a coordinated multi-agent framework. While large language models (LLMs) can parse and reason over unstructured text, applying a single-pass LLM to log synthesis is insufficient due to the scale, structural complexity, and strict semantic constraints of real-world execution data. In particular, execution logs exceed typical context limits, require consistent reasoning over latent states and transitions, and demand explicit control to prevent invalid or hallucinated process flows.

To address these challenges, LogSynthFSM decomposes the synthesis task into specialized agents that operate in a closed-loop pipeline. The overall system is formalized as
\[
\mathcal{D}_{\text{aug}} = \Psi_{\text{FSM}}(\mathcal{L}, A_{\text{parse}}, A_{\text{state}}, A_{\text{schema}}, A_{\text{synth}}, A_{\text{eval}}),
\]
where $\mathcal{L}$ denotes the raw execution logs, and $A_{\text{parse}}$, $A_{\text{state}}$, $A_{\text{schema}}$, $A_{\text{synth}}$, and $A_{\text{eval}}$ correspond to the Execution Parsing, State Discovery, Schema Induction, Relational Synthesis, and Consistency Evaluation agents, respectively.

In each iteration, the Execution Parsing Agent converts raw logs into structured events and extracted parameters. The State Discovery Agent infers latent execution states and transitions, recovering an execution state machine that captures normal system behavior. The Schema Induction Agent then constructs a multi-table relational schema with explicit integrity constraints. Guided by the inferred state machine, the Relational Synthesis Agent generates synthetic relational executions by sampling valid state paths and synthesizing event and parameter records under structural and temporal constraints. The Consistency Evaluation Agent assesses the generated data using constraint validity, distributional similarity, and process fidelity metrics, and provides feedback to guide subsequent iterations. This multi-agent, closed-loop design enables LogSynthFSM to enforce execution semantics, relational consistency, and process validity throughout synthesis, producing high-quality multi-relational synthetic data that cannot be reliably achieved with single-pass LLM generation.

The iteration budget $T$ bounds closed-loop refinement and ensures termination even when candidate batches fail consistency checks. The algorithm retains only batches that satisfy the structural, process, and distributional criteria imposed by the Consistency Evaluation Agent. If no candidate is accepted within the budget, the synthetic component remains empty and downstream training uses only the real relational tables. Empirically, Algorithm~\ref{alg:logsynthfsm} reaches an accepted batch within this budget on all four datasets. This suggests that the FSM and induced schema provide sufficiently strong constraints to keep refinement stable in practice.

\begin{algorithm}[t]
\caption{\textsc{LogSynthFSM}}
\label{alg:logsynthfsm}
\KwIn{Execution logs $\mathcal{L}$; augmentation budget $B$; maximum iterations $T$}
\KwOut{Augmented relational dataset $\mathcal{D}_{\text{aug}}$}

\BlankLine
\textbf{Execution Parsing}\;
$\mathcal{E} \leftarrow A_{\text{parse}}(\mathcal{L})$\;  

\textbf{State Discovery}\;
$\mathcal{F} \leftarrow A_{\text{state}}(\mathcal{E})$\; 

\textbf{Schema Induction}\;
$(\mathcal{R}, \mathcal{D}) \leftarrow A_{\text{schema}}(\mathcal{E}, \mathcal{F})$\; 

\BlankLine
$\mathcal{D}_{\text{best}} \leftarrow \emptyset$\;
\For{$t \leftarrow 1$ \KwTo $T$}{
    $\mathcal{D}_s^{(t)} \leftarrow A_{\text{synth}}(\mathcal{F}, \mathcal{R}, B)$\; 
    $\text{valid}^{(t)} \leftarrow A_{\text{eval}}(\mathcal{D}_s^{(t)}, \mathcal{F}, \mathcal{R})$\; 
    \If{$\text{valid}^{(t)}$}{
        $\mathcal{D}_{\text{best}} \leftarrow \mathcal{D}_s^{(t)}$\;
        \textbf{break}\;
    }
}
\BlankLine
$\mathcal{D}_{\text{aug}} \leftarrow \mathcal{D} \cup \mathcal{D}_{\text{best}}$\; 
\Return $\mathcal{D}_{\text{aug}}$\;
\end{algorithm}

\subsection{Execution Parsing Agent}

The Execution Parsing Agent converts raw, unstructured logs into structured event representations suitable for downstream reasoning. Although logs are textual, they are generated by deterministic program behaviors and follow a template-based structure rather than natural language distributions. Directly reasoning over raw log text obscures execution semantics and introduces unnecessary variability that hinders structure discovery.

We model logs as a sequence of entries $\mathcal{L} = \{\ell_1, \dots, \ell_N\}$, where each entry is generated from a latent template with concrete parameters. The agent decomposes each log entry into a structured event
\[
e_i = (t_i, \theta_i, \tau_i),
\]
where $t_i$ denotes the recovered template, $\theta_i$ represents extracted parameters, and $\tau_i$ is the timestamp. Template extraction reduces lexical noise and exposes stable execution patterns, acting as a compression operator that preserves execution-relevant information.

To ensure consistency and scalability, the agent performs deterministic parsing and normalizes parameters into typed representations, enabling relational reasoning and statistical modeling. By transforming long log streams into compact event sequences, the Execution Parsing Agent also mitigates LLM context limitations, allowing subsequent agents to reason over execution behavior at scale. This structured abstraction forms the foundation for execution state discovery, relational schema induction, and synthetic data generation in LogSynthFSM.

\subsection{State Discovery Agent}

The State Discovery Agent infers latent execution states and transitions from structured event sequences. Although logs are observed as linear sequences of events, they are generated by an underlying execution process governed by unobserved control states. Recovering these latent states is essential for modeling execution semantics and enabling structure-guided data synthesis.

We model the execution process as a finite state machine (FSM)
\[
\mathcal{F} = (\mathcal{S}, \mathcal{A}, \delta),
\]
where $\mathcal{S}$ denotes a finite set of execution states, $\mathcal{A}$ is the alphabet of observed event templates, and $\delta$ is the state transition function. Each event is emitted conditionally on the current latent state, and the FSM itself must be inferred from event traces.

The agent seeks a state abstraction that balances expressiveness and identifiability. Defining states at the template level leads to over-fragmentation, while overly coarse abstractions obscure meaningful execution phases. To address this, the agent groups events into states such that the resulting state sequence approximately satisfies a Markov property:
\[
P(e_{i+1} \mid e_1, \dots, e_i) \approx P(e_{i+1} \mid s_i),
\]
ensuring that inferred states capture the execution history relevant for predicting future behavior while minimizing state complexity.

To operate under context and data constraints, the State Discovery Agent performs iterative, LLM-guided clustering and refinement over compact event summaries, jointly inferring states and transitions to produce a globally coherent FSM. The recovered FSM provides the structural prior for relational schema induction and synthetic data generation in LogSynthFSM, and therefore prioritizes stability and interpretability over maximal granularity.

\subsection{Schema Induction Agent}

The Schema Induction Agent converts the inferred execution state machine into an explicit multi-table relational representation. While execution behavior is inherently relational, linking traces, events, states, and parameters, this structure is implicit in logs and must be induced to enable principled synthesis.

Given an inferred FSM $\mathcal{F}$ and structured events $\{e_i\}$, the agent constructs a relational schema
\[
\mathcal{R} = \{R_{\text{trace}}, R_{\text{event}}, R_{\text{state}}, R_{\text{transition}}, R_{\text{param}}\},
\]
with primary and foreign key constraints encoding execution consistency. The schema is induced to satisfy two properties: \emph{structural faithfulness}, ensuring all executions correspond to valid FSM paths, and \emph{generative sufficiency}, ensuring the schema retains all information required to reproduce execution behavior under the FSM.

By factorizing execution data into normalized relational components, schema induction separates control flow from state-conditional attributes. This factorization constrains the space of admissible synthetic data to structurally valid executions and provides the structural interface between state discovery and relational synthesis in LogSynthFSM.

\subsection{Relational Synthesis Agent}

The Relational Synthesis Agent generates multi-table relational synthetic data conditioned on the execution structure recovered by the preceding agents. Unlike sequence-level augmentation methods, relational synthesis operates over a factorized representation of execution behavior, requiring coordinated generation across multiple relations while preserving global process semantics.

Formally, given an inferred execution state machine $\mathcal{F} = (\mathcal{S}, \mathcal{A}, \delta)$ and relational schema $\mathcal{R}$, the agent synthesizes a set of relational instances
\[
\mathcal{D}_s = \{R_s \mid R \in \mathcal{R}\},
\]
such that all generated tuples satisfy schema constraints and correspond to valid execution paths under $\mathcal{F}$. Generation proceeds by sampling state trajectories from the FSM and conditionally generating event and parameter records according to state-specific distributions.

The synthesis objective is guided by two principles. First, \emph{process validity}: every synthetic execution must correspond to a legal FSM traversal, ensuring structural coherence across relations. Second, \emph{distributional fidelity}: synthetic tuples must match the empirical distributions of state-conditional attributes observed in real executions. By enforcing these principles, the agent enables amplification of rare but valid execution behaviors without introducing structurally invalid or semantically implausible data.

From a modeling perspective, relational synthesis decouples control flow generation from attribute realization, allowing new executions to be composed by recombining valid state paths with realistic parameter instantiations. Within LogSynthFSM, the Relational Synthesis Agent serves as the primary mechanism for expanding the execution dataset while remaining tightly constrained by the discovered execution model.

\begin{table*}
\centering
\caption{Results (mean $\pm$ std across datasets) for anomaly detection utility and synthetic-data fidelity. Higher is better ($\uparrow$) except where noted ($\downarrow$). For C2ST AUC, values closer to 0.5 indicate synthetic traces are harder to distinguish from real traces. Averaged across 25 runs, for all datasets and LLMs.}
\label{tab:main_metrics_avg}
\resizebox{\textwidth}{!}{%
\begin{tabular}{llcccccc}
\toprule
\textbf{Type} & \textbf{Model} &
\textbf{PR-AUC} $\uparrow$ &
\textbf{FPR@95\%R} $\downarrow$ &
\textbf{Rare PR-AUC} $\uparrow$ &
\textbf{TVR} $\uparrow$ &
\textbf{$k$-gram JS} $\downarrow$ &
\textbf{C2ST AUC} (ideal $\approx 0.50$) \\
\midrule
Original & Original Model
& 0.612$\pm$0.041 & 0.138$\pm$0.024 & 0.331$\pm$0.052 & 0.862$\pm$0.033 & 0.214$\pm$0.031 & 0.742$\pm$0.039 \\
\midrule
Classic-Aug & GAugO
& 0.623$\pm$0.038 & 0.132$\pm$0.020 & 0.348$\pm$0.049 & 0.871$\pm$0.029 & 0.203$\pm$0.029 & 0.721$\pm$0.041 \\
\midrule
LLM-Aug & GraphEdit
& 0.655$\pm$0.044 & 0.121$\pm$0.019 & 0.382$\pm$0.057 & 0.889$\pm$0.026 & 0.189$\pm$0.027 & 0.681$\pm$0.036 \\
 & LLM4RGNN
& 0.641$\pm$0.047 & 0.126$\pm$0.021 & 0.369$\pm$0.061 & 0.881$\pm$0.028 & 0.196$\pm$0.030 & 0.694$\pm$0.034 \\
\midrule
Classic-Syn & GraphSMOTE
& 0.628$\pm$0.050 & 0.130$\pm$0.027 & 0.356$\pm$0.063 & 0.903$\pm$0.022 & 0.181$\pm$0.033 & 0.662$\pm$0.029 \\
 & G-Mixup
& 0.619$\pm$0.046 & 0.134$\pm$0.025 & 0.342$\pm$0.058 & 0.898$\pm$0.021 & 0.187$\pm$0.031 & 0.671$\pm$0.033 \\
 & IntraMix
& 0.579$\pm$0.052 & 0.146$\pm$0.029 & 0.301$\pm$0.069 & 0.876$\pm$0.027 & 0.209$\pm$0.036 & 0.708$\pm$0.040 \\
 & GraphAdasyn
& 0.636$\pm$0.048 & 0.127$\pm$0.023 & 0.363$\pm$0.060 & 0.912$\pm$0.018 & 0.176$\pm$0.029 & 0.651$\pm$0.031 \\
 & FG-SMOTE
& 0.631$\pm$0.053 & 0.129$\pm$0.026 & 0.359$\pm$0.067 & 0.906$\pm$0.020 & 0.179$\pm$0.034 & 0.659$\pm$0.028 \\
 & AGMixup
& 0.603$\pm$0.049 & 0.141$\pm$0.028 & 0.328$\pm$0.064 & 0.893$\pm$0.024 & 0.195$\pm$0.035 & 0.676$\pm$0.037 \\
\midrule
LLM-Syn & GAG
& 0.671$\pm$0.043 & 0.115$\pm$0.020 & 0.401$\pm$0.055 & 0.931$\pm$0.016 & 0.162$\pm$0.025 & 0.592$\pm$0.030 \\
 & LLM4NG
& 0.622$\pm$0.058 & 0.129$\pm$0.031 & 0.334$\pm$0.072 & 0.944$\pm$0.014 & 0.155$\pm$0.024 & 0.563$\pm$0.027 \\
 & Mixed-LLM
& 0.684$\pm$0.039 & 0.110$\pm$0.018 & 0.416$\pm$0.052 & 0.948$\pm$0.012 & 0.149$\pm$0.022 & 0.545$\pm$0.021 \\
 & Synthesis-LLM
& 0.677$\pm$0.045 & 0.112$\pm$0.019 & 0.409$\pm$0.058 & 0.952$\pm$0.011 & 0.146$\pm$0.023 & 0.537$\pm$0.020 \\
\midrule
LLM-Syn & GraphMaster
& 0.701$\pm$0.036 & 0.103$\pm$0.017 & 0.438$\pm$0.049 & 0.959$\pm$0.009 & 0.137$\pm$0.021 & 0.524$\pm$0.018 \\
\midrule
\textbf{Ours} & \textbf{LogSynthFSM}
& \textbf{0.742$\pm$0.033} & \textbf{0.087$\pm$0.014} & \textbf{0.511$\pm$0.046} & \textbf{0.985$\pm$0.006} & \textbf{0.091$\pm$0.015} & \textbf{0.503$\pm$0.011} \\
\bottomrule
\end{tabular}%
}
\end{table*}

\subsection{Consistency Evaluation Agent}

The Consistency Evaluation Agent assesses the structural and semantic validity of synthesized relational data and provides feedback to guide iterative refinement. While the Relational Synthesis Agent generates data conditioned on the inferred execution model, generation alone does not guarantee adherence to execution semantics, distributional realism, or process-level coherence. Explicit evaluation is therefore required to prevent drift, hallucination, and structural violations. Formally, given synthesized relational data $\mathcal{D}_s$ under schema $\mathcal{R}$ and execution state machine $\mathcal{F}$, the agent evaluates consistency along three complementary dimensions. \emph{Structural consistency} verifies that all tuples satisfy schema constraints, including primary and foreign key integrity, temporal ordering, and valid state transitions. \emph{Process consistency} ensures that each synthesized execution corresponds to a legal traversal of $\mathcal{F}$, enforcing alignment between relational records and the recovered execution dynamics. \emph{Distributional consistency} compares state-conditional attribute distributions in $\mathcal{D}_s$ to those observed in real data, identifying deviations that indicate unrealistic synthesis.

These evaluations define a constraint-driven acceptance criterion over the synthetic data space, restricting admissible generations to those consistent with both the relational schema and the execution model. From a theoretical perspective, the Consistency Evaluation Agent acts as a projection operator that maps unconstrained generative outputs back onto the feasible set defined by $\mathcal{F}$ and $\mathcal{R}$. Violations detected during evaluation are propagated as feedback signals that inform subsequent synthesis iterations. Within LogSynthFSM, the Consistency Evaluation Agent closes the multi-agent loop by coupling generation with validation. By enforcing explicit structural, process, and distributional constraints throughout synthesis, the agent ensures that synthetic relational data remains grounded in discovered execution semantics, enabling reliable augmentation for downstream anomaly detection.

\section{Experiments}
\subsection{Experimental Setup}
\paragraph{Hardware.}
All experiments were conducted on a single consumer-grade workstation equipped with an AMD Ryzen 9 7900X CPU (12 cores, 24 threads, 4.7 GHz base clock), 32 GB of DDR5 memory, and an NVIDIA RTX 4080 SUPER GPU with 16 GB of VRAM. This setup reflects hardware that is readily accessible to individual researchers and practitioners, demonstrating that LogSynthFSM does not rely on specialized infrastructure or large-scale compute clusters for execution state discovery or relational data synthesis.

\paragraph{Models.}
Our experiments employ open-source large language models to implement the multi-agent components of LogSynthFSM, highlighting the framework’s flexibility and independence from specialized hardware. Autoregressive LLMs are used for the Execution Parsing, State Discovery, Schema Induction, and Consistency Evaluation agents, where reliable semantic understanding and structured reasoning are required. We evaluate several 7B--9B parameter-scale models, including LLaMA-3.1 8B~\cite{touvron2023llama3} (32 layers, 40 attention heads, 4-bit quantization; greedy decoding with temperature 0), Qwen-3 8B~\cite{bai2023qwen} (multilingual, LoRA-enabled; deterministic decoding), DeepSeek-R1 8B~\cite{deepseek2024r1} (NTK-aware; greedy decoding), Gemma-2 9B~\cite{gemma2024} (beam search with width 3), and Mistral 7B~\cite{jiang2023mistral} (grouped-query attention with 8-bit decoding). For simplicity and reproducibility, the same autoregressive model configuration is used across all non-generative agents within a given experiment. The Relational Synthesis Agent uses the same underlying autoregressive model but operates with stochastic decoding to encourage diversity while remaining constrained by the inferred execution state machine and relational schema. Specifically, we use temperature 0.7 with nucleus sampling at $p = 0.9$. All models operate with a maximum context length of 1{,}024 tokens and run in mixed-precision (FP16 or 8-bit) modes where supported. Across all model configurations, LogSynthFSM exhibits stable and consistent performance without reliance on model-specific optimizations or specialized compute infrastructure.

\paragraph{Datasets.}
We evaluate LogSynthFSM on four widely used real-world system log datasets: BGL~\cite{xu2015detecting}, Thunderbird~\cite{xu2015detecting}, Hadoop~\cite{lin2016logclustering}, and OpenStack~\cite{he2016experience}. These datasets are sourced from \textsc{LogHub}~\cite{zhu2019loghub}, a standardized benchmark repository for log analysis that provides diverse logs collected from large-scale production systems. We select these datasets to cover a range of system domains, including high-performance computing (BGL, Thunderbird), distributed data processing frameworks (Hadoop), and cloud infrastructure platforms (OpenStack). Each dataset exhibits distinct characteristics in terms of log volume, template diversity, execution complexity, and anomaly patterns, making them well suited for evaluating the robustness and generality of execution structure discovery and relational synthetic data generation. Importantly, these logs contain rare and complex failure behaviors that are sparsely represented in the original data, providing a realistic setting to assess whether LogSynthFSM can recover latent execution structure and generate structurally valid synthetic data that improves downstream anomaly and bug detection.

\paragraph{Baselines.}
We evaluate LogSynthFSM against a broad and representative set of baselines spanning classical augmentation, graph-based synthesis, and recent LLM-driven approaches. The baselines fall into five categories. First, we include \emph{original training} without any form of data synthesis to establish a lower-bound reference. Second, we consider \emph{classic data augmentation} methods such as GAugO~\cite{liu2021gaug}, which modify existing structures without generating new instances. Third, we evaluate \emph{LLM-based data augmentation} approaches, including GraphEdit~\cite{xia2023graphedit} and LLM4RGNN~\cite{wang2023llm4rgnn}, which leverage language models to edit or enrich existing graph data. Fourth, we include a diverse set of \emph{classic data synthesis} methods that generate new samples through interpolation or resampling, namely GraphSMOTE~\cite{zhao2021graphsmote}, G-Mixup~\cite{verma2019graphmixup}, IntraMix~\cite{wang2022intramix}, GraphADASYN~\cite{ding2021graphadasyn}, FG-SMOTE~\cite{zhao2023fgsmote}, and AGMixup~\cite{liu2022agmixup}. Finally, we compare against \emph{LLM-based data synthesis} methods, including GAG~\cite{zhang2023gag} and LLM4NG~\cite{wang2023llm4ng}, as well as prompt-based synthesis variants (Mixed-LLM and Synthesis-LLM) that isolate the effects of naive LLM generation without explicit structural guidance. For fairness, all baselines operate on execution graphs derived from the same parsed log traces: non-LLM baselines receive directed event graphs with template embeddings as node features, while LLM baselines receive the same traces represented as text-attributed graphs. We additionally include GraphMaster~\cite{du2024graphmaster} as a strong multi-agent LLM synthesis baseline and apply it according to its original text-attributed graph formulation. All methods are evaluated under identical downstream anomaly detection models and evaluation protocols to ensure a controlled and meaningful comparison.

The multi-agent design decomposes the long-horizon synthesis task into verifiable stages, exposes checkable intermediate artifacts, and filters invalid outputs through closed-loop validation. These checks are empirically necessary: removing the Consistency Evaluation Agent lowers PR-AUC from $0.742$ to $0.703$ and increases C2ST AUC from $0.503$ to $0.567$, showing that unchecked errors propagate into less realistic and less useful traces. Replacing the pipeline with a single-pass LLM also degrades downstream and fidelity metrics, indicating that decomposition is functional rather than cosmetic. Detailed prompts and representative agent outputs are provided in the appendix.

\paragraph{Evaluation Metrics.}
We evaluate LogSynthFSM using a unified set of downstream anomaly detection and structure-aware validation metrics to assess both utility and fidelity of the generated data. A graph neural network (GNN)–based anomaly detector is trained on augmented training data and evaluated exclusively on held-out real test logs. We report PR-AUC as the primary metric due to the severe class imbalance in system log anomaly detection, along with FPR at 95\% recall to capture operational performance at high sensitivity. To specifically measure gains on underrepresented behaviors, we additionally report PR-AUC (or F1) on a rare-event slice of the test set defined by infrequent FSM transitions or low-occupancy states. To validate the realism and structural correctness of generated data, we report transition validity rate, which measures adherence to the inferred execution state machine, $k$-gram path Jensen--Shannon divergence to quantify process-level similarity between real and synthetic executions, and a trace-level classifier two-sample test (C2ST) AUC to evaluate the distinguishability of synthetic and real traces. Together, these metrics follow evaluation principles by jointly measuring predictive performance, robustness under imbalance, and faithfulness to discovered execution structure, ensuring that observed improvements arise from meaningful execution-aware synthesis rather than increased data volume alone.

\paragraph{Multi-Agent Collusion and Privacy.}
Multi-agent generation pipelines are vulnerable to both collusion and memorization, where agents implicitly optimize for internal acceptance criteria or reproduce training traces instead of generating faithful but novel executions. LogSynthFSM mitigates these risks through explicit architectural separation and privacy-aware evaluation. Agents operate under role-specific prompts with disjoint objectives, and internal validation is restricted to hard execution constraints (FSM-valid transitions, template admissibility, and relational key integrity), preventing evaluators from accepting outputs that merely appear plausible or echo seen data. Privacy and memorization leakage are addressed at the metric level: all downstream utility is measured on held-out real test data, while trace-level C2ST AUC and $k$-gram path divergence act as adversarial checks that expose near-duplicate traces and overfitting to training patterns. Structure-grounded metrics such as transition validity rate further ensure that improvements arise from execution-faithful generalization rather than memorized sequences. Together, these safeguards limit both collusive feedback loops and unintended memorization, supporting privacy-conscious synthetic data generation for anomaly detection.

\begin{table*}[t]
\centering
\caption{Ablation study for LogSynthFSM (mean $\pm$ std across datasets). Higher is better ($\uparrow$) except where noted ($\downarrow$). For C2ST AUC, values closer to 0.50 indicate synthetic traces are harder to distinguish from real traces. All downstream scores are computed on a real-only test set.}
\label{tab:ablation_logsynthfsm}
\resizebox{\textwidth}{!}{%
\begin{tabular}{llcccccc}
\toprule
\textbf{Ablation Type} & \textbf{Variant} &
\textbf{PR-AUC} $\uparrow$ &
\textbf{FPR@95\%R} $\downarrow$ &
\textbf{Rare PR-AUC} $\uparrow$ &
\textbf{TVR} $\uparrow$ &
\textbf{$k$-gram JS} $\downarrow$ &
\textbf{C2ST AUC} (ideal $\approx 0.50$) \\
\midrule

\multirow{5}{*}{Agent removal}
& w/o Execution Parsing Agent
& 0.668$\pm$0.041 & 0.121$\pm$0.018 & 0.372$\pm$0.055 & 0.901$\pm$0.016 & 0.173$\pm$0.023 & 0.558$\pm$0.021 \\
& w/o State Discovery Agent (no FSM; sequence-only synthesis)
& 0.643$\pm$0.048 & 0.129$\pm$0.021 & 0.341$\pm$0.061 & 0.714$\pm$0.042 & 0.206$\pm$0.031 & 0.621$\pm$0.026 \\
& w/o Schema Induction Agent (single-table flattening)
& 0.684$\pm$0.039 & 0.116$\pm$0.019 & 0.389$\pm$0.049 & 0.918$\pm$0.014 & 0.162$\pm$0.022 & 0.545$\pm$0.019 \\
& w/o Relational Synthesis Agent (no synthetic data; real-only training)
& 0.612$\pm$0.041 & 0.138$\pm$0.024 & 0.331$\pm$0.052 & 0.862$\pm$0.033 & 0.214$\pm$0.031 & 0.742$\pm$0.039 \\
& w/o Consistency Evaluation Agent (no closed-loop validation)
& 0.703$\pm$0.044 & 0.109$\pm$0.018 & 0.433$\pm$0.058 & 0.938$\pm$0.019 & 0.149$\pm$0.026 & 0.567$\pm$0.023 \\
\midrule

\multirow{6}{*}{Design choices}
& Replace multi-agent with single LLM pass (no decomposition)
& 0.691$\pm$0.046 & 0.114$\pm$0.020 & 0.412$\pm$0.061 & 0.911$\pm$0.021 & 0.165$\pm$0.027 & 0.575$\pm$0.024 \\
& No iterative refinement (1 iteration only)
& 0.712$\pm$0.038 & 0.104$\pm$0.016 & 0.448$\pm$0.053 & 0.962$\pm$0.010 & 0.128$\pm$0.020 & 0.523$\pm$0.016 \\
& No rare-event amplification ($\pi$ uniform over states)
& 0.721$\pm$0.035 & 0.098$\pm$0.015 & 0.421$\pm$0.047 & 0.979$\pm$0.007 & 0.103$\pm$0.016 & 0.509$\pm$0.012 \\
& Replace FSM guidance with unigram template sampling
& 0.664$\pm$0.050 & 0.125$\pm$0.022 & 0.352$\pm$0.064 & 0.781$\pm$0.031 & 0.192$\pm$0.028 & 0.612$\pm$0.025 \\
& Remove parameter conditioning (template-only synthesis)
& 0.707$\pm$0.039 & 0.107$\pm$0.017 & 0.436$\pm$0.052 & 0.970$\pm$0.009 & 0.121$\pm$0.018 & 0.526$\pm$0.015 \\
& Use strict deterministic decoding for synthesis (Temp=0)
& 0.719$\pm$0.036 & 0.101$\pm$0.015 & 0.405$\pm$0.050 & 0.984$\pm$0.006 & 0.097$\pm$0.015 & 0.531$\pm$0.014 \\
\midrule

\textbf{Ours}
& \textbf{LogSynthFSM (full)}
& \textbf{0.742$\pm$0.033} & \textbf{0.087$\pm$0.014} & \textbf{0.511$\pm$0.046} & \textbf{0.985$\pm$0.006} & \textbf{0.091$\pm$0.015} & \textbf{0.503$\pm$0.011} \\
\bottomrule
\end{tabular}%
}
\end{table*}

\subsection{Discussion}

\begin{figure}[t]
    \centering
    \includegraphics[width=\columnwidth]{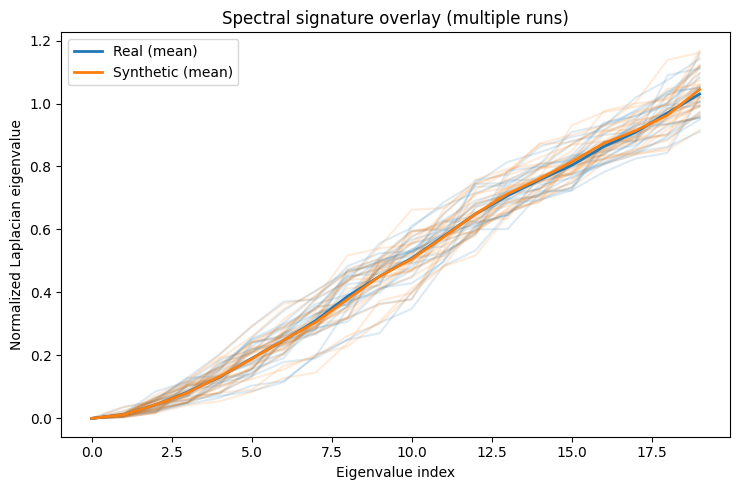}
    \caption{Spectral signature overlay of execution graphs derived from real and synthetic traces. The close alignment of normalized Laplacian eigenvalues indicates preservation of global execution structure under synthesis, while visible variance reflects stochastic generation and rare-path amplification.}
    \label{fig:spectral}
\end{figure}

\paragraph{Main Results. }
The improvements achieved by LogSynthFSM in Table~\ref{tab:main_metrics_avg} are driven by its explicit modeling of execution structure and its constraint-aware multi-agent synthesis pipeline. Compared to the strongest LLM-based baseline, GraphMaster, LogSynthFSM improves PR-AUC from $0.701$ to $0.742$ and reduces FPR at 95\% recall from $0.103$ to $0.087$, indicating more reliable anomaly detection at high-sensitivity operating points. These gains are particularly pronounced for rare behaviors: rare-event PR-AUC increases from $0.438$ to $0.511$, confirming that state-aware rare-path amplification substantially improves coverage of underrepresented but valid executions. Structural fidelity metrics further explain these downstream improvements. LogSynthFSM achieves a transition validity rate of $0.985$, compared to $0.959$ for GraphMaster and $0.931$ for GAG, demonstrating that the inferred execution state machine provides stronger process constraints during synthesis. Similarly, the reduction in $k$-gram path Jensen--Shannon divergence from $0.137$ (GraphMaster) to $0.091$ indicates closer alignment with real execution dynamics at the multi-step path level. Finally, the near-ideal C2ST AUC of $0.503$ shows that synthetic traces generated by LogSynthFSM are difficult to distinguish from real traces, suggesting that performance gains are not due to distributional artifacts or memorization. Together, these results highlight that the primary source of improvement is not the use of larger language models alone, but the architectural integration of execution parsing, state discovery, schema induction, and constraint-aware relational synthesis within a closed-loop multi-agent framework, which jointly ensures structural correctness, distributional realism, and measurable gains in real-world anomaly detection. To avoid circularity, we distinguish internal structural validity from external utility: TVR is measured against the inferred FSM, whereas $k$-gram JS, C2ST AUC, PR-AUC, and Rare PR-AUC are computed against real data, with PR-AUC metrics evaluated only on held-out real test traces unseen by the FSM. Thus, a self-consistent but meaningless FSM would not improve real anomaly detection, while the Rare PR-AUC gain from $0.438$ to $0.511$ over the strongest baseline provides non-circular evidence that the recovered structure captures genuine execution behavior.

\begin{figure*}[t]
    \centering
    \includegraphics[width=\textwidth]{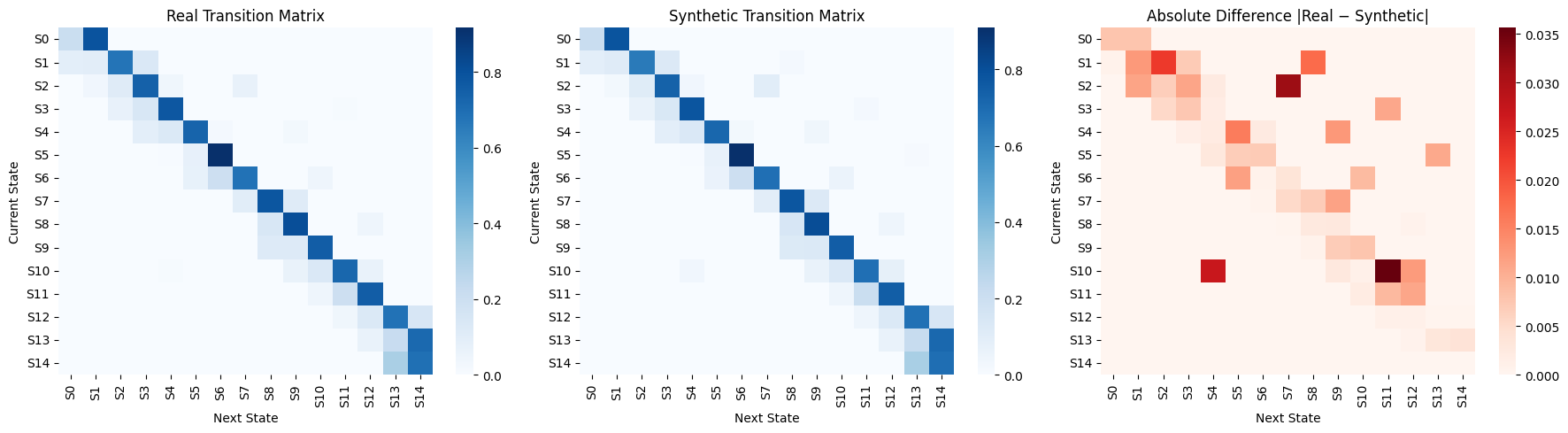}
    \caption{Comparison of real and synthetic execution transition matrices. The difference highlights where rare but valid transitions are amplified while preserving overall execution structure.}
    \label{fig:transition_matrix}
\end{figure*}

\paragraph{Structural Fidelity Analysis.}
Figures~\ref{fig:spectral} and~\ref{fig:transition_matrix} provide complementary evidence that the relational databases generated by LogSynthFSM preserve both the global graph structure and the local execution semantics encoded by the discovered FSM. The transition matrix comparison in Figure~\ref{fig:transition_matrix} shows that synthetic traces adhere closely to the transition support of the real execution graph, with differences concentrated on low-probability but valid transitions rather than the introduction of spurious edges. This pattern indicates that the synthesis process respects hard FSM constraints while selectively amplifying rare execution paths, which is consistent with the observed improvements in rare-event detection. At a global level, the spectral signature overlay in Figure~\ref{fig:spectral} demonstrates strong alignment between the normalized Laplacian eigenvalues of real and synthetic execution graphs, implying preservation of higher-order structural properties such as connectivity, flow bottlenecks, and state centrality. Because spectral characteristics capture aggregate path and connectivity patterns beyond individual transitions, this alignment provides evidence that the synthesized relational data maintains the same execution geometry as the real system. Together, these figures show that LogSynthFSM does not merely reproduce local transition frequencies, but generates relational data that is structurally consistent with the inferred FSM at both the transition and graph levels, supporting it as a structure-preserving data generation framework for anomaly detection.

\begin{table}[t]
\centering
\caption{Model transferability across LLM backbones. Runtime reports end-to-end local synthesis and training time; API runtimes are omitted because hardware is not controlled.}
\label{tab:model_transferability}
\setlength{\tabcolsep}{3pt}
\scriptsize
\resizebox{\columnwidth}{!}{%
\begin{tabular}{lccc}
\toprule
\textbf{LLM Backbone} & \textbf{Accuracy (\%)} $\uparrow$ & \textbf{PR-AUC} $\uparrow$ & \textbf{Runtime (min)} $\downarrow$ \\
\midrule
LLaMA-3.1 8B & 78.4 $\pm$ 0.9 & 0.739 $\pm$ 0.046 & 46.2 $\pm$ 2.1 \\
Qwen-3 8B & 77.6 $\pm$ 1.1 & 0.734 $\pm$ 0.057 & 44.8 $\pm$ 1.9 \\
DeepSeek-R1 8B & 79.1 $\pm$ 1.0 & 0.744 $\pm$ 0.026 & 49.5 $\pm$ 2.4 \\
Gemma-2 9B & 78.0 $\pm$ 1.2 & 0.737 $\pm$ 0.028 & 52.3 $\pm$ 2.6 \\
Mistral 7B & 76.9 $\pm$ 1.3 & 0.729 $\pm$ 0.048 & 41.7 $\pm$ 1.8 \\
\midrule
GPT-4.1 Nano (API) & 79.6 $\pm$ 0.8 & 0.747 $\pm$ 0.025 & -- \\
GPT-4.1 Mini (API) & 80.3 $\pm$ 0.7 & 0.752 $\pm$ 0.044 & -- \\
GPT-4.1 (API) & 80.8 $\pm$ 0.6 & 0.754 $\pm$ 0.044 & -- \\
\bottomrule
\end{tabular}%
}
\end{table}

\subsection{Ablation Study}
\paragraph{Ablation Analysis.}
The ablation results in Table~\ref{tab:ablation_logsynthfsm} highlight the complementary roles of LogSynthFSM’s architectural components and design choices. Removing the Execution Parsing Agent degrades PR-AUC from $0.742$ to $0.668$ and increases $k$-gram divergence from $0.091$ to $0.173$, indicating that inaccurate extraction of execution primitives propagates structural noise throughout the synthesis pipeline. The most severe degradation occurs when the State Discovery Agent is removed: eliminating FSM guidance reduces transition validity rate from $0.985$ to $0.714$ and rare-event PR-AUC from $0.511$ to $0.341$, confirming that explicit state modeling is essential for preserving valid execution paths and improving coverage of infrequent behaviors. Flattening the relational schema into a single table similarly reduces both downstream utility and realism, demonstrating that multi-table structure is necessary for capturing cross-entity dependencies exploited by the GNN classifier. The removal of the Consistency Evaluation Agent yields a noticeable drop in PR-AUC (from $0.742$ to $0.703$) and an increase in C2ST AUC (from $0.503$ to $0.567$), showing that closed-loop validation is critical for preventing structurally plausible but distributionally distinguishable generations. Design-level ablations further reinforce these findings: replacing the multi-agent pipeline with a single-pass LLM or disabling iterative refinement consistently harms rare-event performance and increases structural divergence, while removing rare-event amplification disproportionately impacts rare PR-AUC despite maintaining high transition validity. Finally, deterministic decoding slightly improves validity metrics but reduces downstream utility, indicating a trade-off between diversity and strict constraint satisfaction. Overall, these results demonstrate that LogSynthFSM’s gains arise from the joint interaction of execution-aware decomposition, FSM-constrained synthesis, relational structure, and iterative consistency enforcement, rather than any single component in isolation.

\paragraph{Benefit of Relational Generation. }
The relational schema serves two functions that direct text-attributed graph generation cannot replicate. First, it enforces referential integrity across traces, events, states, transitions, and parameters during synthesis, preventing structurally invalid combinations before graph construction. Second, it separates static execution structure from dynamic instance data, allowing the GNN to exploit cross-entity dependencies through message passing over the derived graph. The ablation confirms this effect: flattening the schema to a single table before graph construction reduces PR-AUC from $0.742$ to $0.684$, showing that graph quality depends on the integrity of the underlying relational representation.

\paragraph{Model Transferability. }
Table~\ref{tab:model_transferability} demonstrates that LogSynthFSM is largely model-agnostic, achieving consistent downstream accuracy across a range of open-source and API-based LLM backbones. Performance varies by less than three percentage points across locally hosted models, indicating that improvements are driven by the execution-aware synthesis framework rather than reliance on a specific language model. Importantly, competitive performance is obtained using compact 7--9B open-source models with end-to-end runtimes under one hour on a single consumer GPU, highlighting that LogSynthFSM does not require specialized hardware or large-scale compute infrastructure. While API-based models yield marginal accuracy gains, they are not essential for strong performance, reinforcing the practicality and accessibility of the proposed system for real-world deployment.

\paragraph{Runtime and Token Overhead. }
LogSynthFSM incurs higher runtime than classical baselines because it performs iterative parsing, state discovery, schema induction, synthesis, and consistency evaluation. As shown in Table~\ref{tab:model_transferability}, local runs take 41.7--52.3 minutes on a single RTX 4080 SUPER GPU, while classical baselines finish in under three minutes but achieve weaker fidelity and detection performance. Token cost remains bounded by the 1{,}024-token context cap used by each agent, so API usage scales linearly with log volume and agent calls. At current GPT-4.1-series rates, this corresponds to USD 2.00/8.00, 0.40/1.60, and 0.10/0.40 per 1M input/output tokens for GPT-4.1, GPT-4.1-mini, and GPT-4.1-nano, respectively.

\subsection{Case Study}
We illustrate the behavior of LogSynthFSM using a representative run on the \textsc{Hadoop} log dataset. Starting from raw, unstructured log lines, the execution parsing agent first extracts normalized templates and parameter fields, producing a structured event stream. The state discovery agent then infers a compact FSM that captures dominant execution phases such as job initialization, task scheduling, execution, and cleanup, along with low-frequency transitions corresponding to failure recovery and retry paths. Based on this FSM, the schema induction agent constructs a multi-relational database linking traces, events, states, transitions, and parameters, enabling explicit representation of control flow and temporal dependencies. The relational synthesis agent subsequently generates additional synthetic traces by sampling state-consistent paths, selectively amplifying rare but valid recovery transitions while preserving overall transition support. Finally, the consistency evaluation agent filters and refines generated data to ensure high transition validity and low distributional drift. When training a downstream GNN anomaly detector on the augmented relational data, the model exhibits improved sensitivity to subtle failure patterns, such as repeated task reassignments and delayed state exits, which are underrepresented in the original logs. This example highlights how LogSynthFSM transforms raw logs into execution-aware relational data and leverages structure-preserving synthesis to improve anomaly detection in a realistic system setting.

\section{Conclusions}
We presented \textsc{LogSynthFSM}, a structure-aware multi-agent framework for generating execution-consistent relational synthetic data from unstructured system logs. By discovering latent execution states and using the inferred FSM as a generative prior, LogSynthFSM preserves process semantics and global execution structure while amplifying rare but valid behaviors. Extensive experiments and ablations show consistent improvements in downstream anomaly detection without sacrificing structural fidelity or realism. Our model-agnostic design and ability to run efficiently on consumer-grade hardware make LogSynthFSM practical for real-world deployment, and point toward a promising direction for execution-aware synthetic data generation in complex software systems.

\bibliographystyle{ACM-Reference-Format}
\bibliography{sample-base}

\appendix

\section{Appendix}
This appendix provides the information required to reproduce the experiments presented in this work on \textbf{LogSynthFSM}. LogSynthFSM formulates execution-aware synthetic data generation as a structured multi-agent process, in which specialized LLM-based agents collaboratively discover latent execution structure from raw logs and generate relational synthetic data under explicit process and schema constraints. Rather than relying on single-pass generation, the system employs iterative coordination among agents responsible for execution parsing, state discovery, schema induction, relational synthesis, and consistency evaluation.

Each experiment is driven by natural language prompts that implement three core stages: (1) extracting execution templates, parameters, and candidate state transitions from raw logs, (2) inducing an execution state machine and multi-table relational schema, and (3) synthesizing relational traces guided by the inferred FSM with iterative validation against structural and distributional constraints.

This appendix includes:
\begin{enumerate}
    \item The exact prompts used by each agent in the LogSynthFSM pipeline, including execution parsing, state discovery, schema induction, synthesis, and consistency evaluation.
    \item A complete example run on a representative log dataset, illustrating how agent interactions and constraint-based feedback lead to stable convergence and structurally valid relational synthetic data.
\end{enumerate}

\subsection{Execution Parsing Agent}
The Execution Parsing Agent is responsible for transforming raw OpenStack log text into a normalized, machine-readable event stream that can be used by later agents to recover execution structure. Concretely, it (i) groups raw lines into trace-aligned sequences (when trace or request identifiers are present), (ii) extracts stable log templates by replacing variable fields such as timestamps, UUIDs, IPs, and numeric values with typed placeholders, and (iii) isolates the corresponding parameters into key–value records while preserving their original values. For each entry, it assigns consistent template and event identifiers and retains metadata such as timestamp and component (for example, \texttt{nova}, \texttt{neutron}, \texttt{keystone}) to maintain ordering and provenance. This step deliberately avoids inferring states or anomalies; instead, it produces a clean, loss-minimized representation of executions where recurring behaviors become comparable across traces, enabling reliable FSM discovery and multi-table relational schema induction downstream.

\begin{tcolorbox}[
    title={Execution Parsing Agent -- Example Output (OpenStack)},
    colback=white,
    colframe=black,
    fonttitle=\bfseries,
    breakable,
    enhanced,
    clip upper,
    clip lower
]

\small
You are an execution parsing agent operating on raw system logs from the OpenStack cloud platform.
Your task is to convert unstructured log lines into a structured event representation suitable for downstream execution analysis.

Given a sequence of raw log entries, perform the following steps:

\begin{enumerate}
    \item Identify and extract a \emph{log template} for each entry by replacing variable fields (e.g., timestamps, UUIDs, IP addresses, numeric identifiers) with typed placeholders.
    \item Extract all variable parameters associated with each template, preserving their original values and semantic roles.
    \item Assign each parsed entry a unique \texttt{event\_id} and record its original \texttt{timestamp} and \texttt{component} (e.g., \texttt{nova}, \texttt{neutron}, \texttt{keystone}) when available.
    \item Preserve the original ordering of events within each execution trace.
\end{enumerate}

Output the result as a structured table with the following columns:
\begin{center}
\texttt{(trace\_id, event\_id, template\_id,}\\
\texttt{\ \ template\_text, parameters, timestamp, component)}
\end{center}

Do \emph{not} infer execution states, transitions, or anomalies at this stage.
Focus strictly on faithful parsing and normalization of the raw log content while preserving all information required for later execution structure discovery.

Ensure that:
\begin{itemize}
    \item Templates are consistent across similar log messages.
    \item Parameter types are correctly identified (e.g., integer, string, UUID, IP).
    \item No log entries are dropped or reordered.
\end{itemize}
\end{tcolorbox}

\begin{tcolorbox}[
    title={Execution Parsing Agent -- Example Output (OpenStack)},
    colback=white,
    colframe=black,
    fonttitle=\bfseries,
    breakable,
    enhanced,
    clip upper,
    clip lower
]
\small
\textbf{Input (raw log lines):}
\begin{lstlisting}
2021-06-01 10:12:43.512 INFO nova.compute.manager [req-7c9a]
    Starting instance e3f1a2c4-9d2a-4f7e-bd9a-12ac34ef56ab
2021-06-01 10:12:44.018 ERROR nova.compute.manager [req-7c9a]
    Failed to allocate network for instance e3f1a2c4-9d2a-4f7e-bd9a-12ac34ef56ab
\end{lstlisting}

\vspace{0.4em}
\textbf{Parsed structured output:}
\begin{lstlisting}
trace_id | event_id | template_id | template_text
------------------------------------------------
T42      | E1       | TMP_017     | Starting instance <UUID>
T42      | E2       | TMP_031     | Failed to allocate network for instance <UUID>

parameters:
E1: { uuid: "e3f1a2c4-9d2a-4f7e-bd9a-12ac34ef56ab" }
E2: { uuid: "e3f1a2c4-9d2a-4f7e-bd9a-12ac34ef56ab" }

timestamps:
E1: 2021-06-01 10:12:43.512
E2: 2021-06-01 10:12:44.018

component:
E1: nova.compute.manager
E2: nova.compute.manager
\end{lstlisting}
\end{tcolorbox}

\subsection{State Discovery Agent}
The State Discovery Agent takes the structured, template-normalized event stream produced by the Execution Parsing Agent and infers a latent execution finite-state machine that summarizes how traces evolve over time. It clusters templates into a small set of semantically coherent states (for example, request initialization, resource lookup, scheduling, network allocation, compute spawn, success finalization, and error handling) using evidence from co-occurrence patterns, component metadata, and consistent temporal ordering across traces. It then induces directed transitions by aggregating observed state-to-state jumps within each trace, estimating transition frequencies or probabilities and explicitly retaining low-frequency but valid paths such as retries and recovery loops. The agent outputs three artifacts: (i) a compact state inventory with human-interpretable descriptions, (ii) a transition table that encodes the FSM dynamics, and (iii) a mapping from each template to its assigned state. This representation converts raw sequential logs into an explicit process model that can be enforced during relational synthesis, ensuring generated multi-table traces respect feasible execution structure rather than only matching local sequence statistics.

\begin{tcolorbox}[
    title={State Discovery Agent -- Initial Prompt (OpenStack)},
    colback=white,
    colframe=black,
    fonttitle=\bfseries,
    breakable,
    enhanced,
    clip upper,
    clip lower
]
\small
You are a state discovery agent operating on structured execution data produced by the Execution Parsing Agent.
Your task is to infer a latent execution state machine (FSM) that explains how software executions evolve over time.

\textbf{Input.}
You are given a table of parsed log events with the following fields:
\begin{center}
\texttt{(trace\_id, event\_id, template\_id, template\_text, parameters, timestamp, component)}
\end{center}
Each trace corresponds to a single execution instance, and events are ordered by timestamp within each trace.

\textbf{Objective.}
Infer a compact, interpretable FSM that captures recurring execution phases and valid transitions between them.
States should represent semantically meaningful execution stages (e.g., request initialization, resource allocation,
execution, cleanup, error handling), rather than individual log templates.

\textbf{Instructions.}
\begin{enumerate}
    \item Group events across traces into candidate execution states based on shared templates, components,
    temporal proximity, and co-occurrence patterns.
    \item Assign each event to exactly one latent state.
    \item Infer directed transitions between states based on observed event orderings within traces.
    \item Estimate transition frequencies or probabilities from the data.
    \item Identify rare but valid transitions (e.g., retry, recovery, failure-handling paths) and retain them.
\end{enumerate}

\textbf{Output.}
Produce:
\begin{itemize}
    \item A set of states with unique \texttt{state\_id} and brief semantic descriptions.
    \item A transition list of the form \texttt{(source\_state, target\_state, count, probability)}.
    \item A mapping from \texttt{template\_id} to \texttt{state\_id}.
\end{itemize}

Do \emph{not} generate synthetic data or modify the input.
Do \emph{not} collapse all events into a single linear sequence.
Favor a small number of reusable states that generalize across traces while preserving valid execution paths.

Ensure that the inferred FSM:
\begin{itemize}
    \item Explains the majority of observed traces.
    \item Avoids spurious transitions not supported by the data.
    \item Preserves rare but legitimate execution behaviors.
\end{itemize}
\end{tcolorbox}

\begin{tcolorbox}[
    title={State Discovery Agent -- Example Output (OpenStack)},
    colback=white,
    colframe=black,
    fonttitle=\bfseries,
    breakable,
    enhanced,
    clip upper,
    clip lower
]
\small
\textbf{Inferred states (semantic summaries):}
\begin{itemize}
    \item \texttt{S0: RequestInit} -- request received, context established, auth/headers parsed
    \item \texttt{S1: ResourceLookup} -- querying service catalogs, DB reads, fetching instance/network metadata
    \item \texttt{S2: Scheduling} -- selecting host, placement decisions, filtering candidates
    \item \texttt{S3: NetworkAlloc} -- neutron port creation, network attachment, IP allocation
    \item \texttt{S4: SpawnCompute} -- instance spawn, image/volume attach, hypervisor operations
    \item \texttt{S5: SuccessFinalize} -- state update, notifications, cleanup after success
    \item \texttt{S6: ErrorHandle} -- failures, exceptions, rollback, cleanup after error
    \item \texttt{S7: RetryRecovery} -- retry scheduling, re-attempting network/compute steps
\end{itemize}

\vspace{0.5em}
\textbf{Transition list (counts and probabilities):}
\begin{verbatim}
(source_state, target_state, count, probability)
(S0, S1, 1820, 0.92)
(S0, S6,  160, 0.08)

(S1, S2, 1655, 0.88)
(S1, S6,  145, 0.08)
(S1, S7,   60, 0.03)
(S1, S5,   20, 0.01)

(S2, S3, 1520, 0.90)
(S2, S7,   85, 0.05)
(S2, S6,   80, 0.05)

(S3, S4, 1370, 0.89)
(S3, S7,  110, 0.07)
(S3, S6,   60, 0.04)

(S4, S5, 1325, 0.94)
(S4, S6,   55, 0.04)
(S4, S7,   25, 0.02)

(S7, S2,  135, 0.41)
(S7, S3,  110, 0.33)
(S7, S4,   45, 0.14)
(S7, S6,   40, 0.12)

(S6, S5,  210, 1.00)
\end{verbatim}

\vspace{0.5em}
\textbf{Template-to-state mapping (excerpt):}
\begin{verbatim}
template_id  -> state_id
TMP_002 -> S0 (Received request <REQ_ID> from <IP>)
TMP_005 -> S0 (Authenticated token <UUID>)
TMP_011 -> S1 (DB query: SELECT ... WHERE instance=<UUID>)
TMP_018 -> S2 (Selected host <HOST> for instance <UUID>)
TMP_021 -> S2 (Placement claim succeeded for <UUID>)
TMP_031 -> S3 (Created port <UUID> on network <UUID>)
TMP_034 -> S3 (Allocated IP <IP> for port <UUID>)
TMP_041 -> S4 (Spawning instance <UUID> on host <HOST>)
TMP_046 -> S4 (Attached volume <UUID> to instance <UUID>)
TMP_052 -> S5 (Instance <UUID> ACTIVE)
TMP_061 -> S6 (Exception: <ERR> while processing <UUID>)
TMP_067 -> S7 (Retrying operation for instance <UUID>)
\end{verbatim}

\vspace{0.3em}
\textbf{Notes.} Rare but valid transitions are preserved through \texttt{S7} (retry/recovery), which captures
re-entrance into scheduling, network allocation, and spawn phases without introducing unsupported edges.
\end{tcolorbox}

\subsection{Schema Induction Agent}
The Schema Induction Agent converts the discovered execution structure into a normalized multi-table relational database design that can represent both the static process model and the dynamic per-trace realizations. Given the FSM (states and transition statistics) and the template-to-state assignments, it identifies the core entities that must be stored separately (for example, traces, events, templates, states, transitions, and event parameters), then defines primary keys and foreign-key relationships to enforce referential integrity and reconstructability. This separation is important: static tables such as \texttt{States}, \texttt{Transitions}, and \texttt{StateTemplateMap} encode the execution constraints that guide synthesis, while instance tables such as \texttt{Traces}, \texttt{Events}, and \texttt{EventParams} store observed or generated executions without duplicating global structure. By avoiding single-table flattening, the induced schema preserves relational dependencies among parameters and components, supports constraint checking (for example, whether emitted templates are admissible for a state and whether consecutive states follow valid transitions), and provides a clean interface for downstream GNN-based anomaly detection on execution graphs derived from the relational data.

\begin{tcolorbox}[
    title={Schema Induction Agent -- Initial Prompt (OpenStack)},
    colback=white,
    colframe=black,
    fonttitle=\bfseries,
    breakable,
    enhanced,
    clip upper,
    clip lower
]
\small
You are a schema induction agent operating on an inferred execution state machine (FSM)
and a template-to-state mapping derived from system logs.

\textbf{Input.}
You are given:
\begin{itemize}
    \item A set of execution states with semantic descriptions.
    \item A directed transition table of the form
    \texttt{(source\_state, target\_state, count, probability)}.
    \item A mapping from \texttt{template\_id} to \texttt{state\_id}.
    \item Parsed event records with fields:
    \begin{center}
    \texttt{(trace\_id, event\_id, template\_id, parameters, timestamp, component)}
    \end{center}
\end{itemize}

\textbf{Objective.}
Induce a normalized multi-table relational schema that faithfully represents executions,
their states, transitions, and parameters, while avoiding redundancy and preserving
referential integrity.

\textbf{Instructions.}
\begin{enumerate}
    \item Identify core entities that should correspond to relational tables
    (e.g., traces, events, states, transitions, parameters).
    \item Define primary keys and foreign-key relationships between tables.
    \item Ensure that each execution trace can be reconstructed by joining tables.
    \item Separate static execution structure (states, transitions) from dynamic
    execution instances (traces, events).
    \item Do not flatten all information into a single table.
\end{enumerate}

\textbf{Output.}
Produce:
\begin{itemize}
    \item A list of tables with schema definitions
    \texttt{TableName(columns, primary key, foreign keys)}.
    \item A brief description of each table’s role in representing execution behavior.
\end{itemize}

\textbf{Constraints.}
\begin{itemize}
    \item The schema must be compatible with FSM-guided synthetic data generation.
    \item All foreign-key references must be valid.
    \item The schema should generalize across traces and not encode trace-specific artifacts.
\end{itemize}

Do \emph{not} generate synthetic data.
Do \emph{not} infer new states or transitions.
Focus solely on relational structure.
\end{tcolorbox}

\begin{tcolorbox}[
    title={Schema Induction Agent -- Example Output (OpenStack)},
    colback=white,
    colframe=black,
    fonttitle=\bfseries,
    breakable,
    enhanced,
    clip upper,
    clip lower
]
\small
\textbf{Induced relational schema (normalized, execution-aware):}

\begin{itemize}
    \item \texttt{Traces(trace\_id, start\_ts, end\_ts, service, is\_anomalous)}\\
    \emph{PK:} \texttt{trace\_id}. \emph{Role:} One row per execution instance; stores coarse metadata.

    \item \texttt{Events(trace\_id, event\_id, ts, component, template\_id, state\_id)}\\
    \emph{PK:} \texttt{(trace\_id, event\_id)}. \emph{FKs:} \texttt{trace\_id} $\rightarrow$ \texttt{Traces}, \texttt{template\_id} $\rightarrow$ \texttt{Templates}, \texttt{state\_id} $\rightarrow$ \texttt{States}.\\
    \emph{Role:} Event stream per trace with explicit state assignment for FSM reasoning.

    \item \texttt{Templates(template\_id, template\_text)}\\
    \emph{PK:} \texttt{template\_id}. \emph{Role:} Canonical template definitions for normalization.

    \item \texttt{States(state\_id, state\_name, state\_desc)}\\
    \emph{PK:} \texttt{state\_id}. \emph{Role:} Latent execution phases (semantic clusters of templates).

    \item \texttt{Transitions(src\_state, dst\_state, count, prob)}\\
    \emph{PK:} \texttt{(src\_state, dst\_state)}. \emph{FKs:} \texttt{src\_state} $\rightarrow$ \texttt{States}, \texttt{dst\_state} $\rightarrow$ \texttt{States}.\\
    \emph{Role:} FSM dynamics (empirical transition statistics used as a generative prior).

    \item \texttt{StateTemplateMap(state\_id, template\_id)}\\
    \emph{PK:} \texttt{(state\_id, template\_id)}. \emph{FKs:} \texttt{state\_id} $\rightarrow$ \texttt{States}, \texttt{template\_id} $\rightarrow$ \texttt{Templates}.\\
    \emph{Role:} Defines which templates are admissible emissions for each state.

    \item \texttt{EventParams(trace\_id, event\_id, key, value, value\_type)}\\
    \emph{PK:} \texttt{(trace\_id, event\_id, key)}. \emph{FKs:} \texttt{(trace\_id, event\_id)} $\rightarrow$ \texttt{Events}.\\
    \emph{Role:} Key--value parameter storage (e.g., UUIDs, hostnames, IPs), separated to avoid sparsity.

    \item \texttt{TraceSummary(trace\_id, n\_events, n\_states, duration\_ms, n\_retries)}\\
    \emph{PK:} \texttt{trace\_id}. \emph{FK:} \texttt{trace\_id} $\rightarrow$ \texttt{Traces}.\\
    \emph{Role:} Optional derived features for downstream classifiers and monitoring.
\end{itemize}

\vspace{0.5em}
\textbf{Join property (reconstructability).}
A trace’s ordered execution can be reconstructed by joining \texttt{Traces} $\bowtie$ \texttt{Events} $\bowtie$ \texttt{Templates} (and optionally \texttt{EventParams}) and sorting by \texttt{ts}.
FSM-level constraints are represented independently via \texttt{States}, \texttt{Transitions}, and \texttt{StateTemplateMap}, enabling constraint-aware synthesis and validation.

\end{tcolorbox}

\subsection{Relational Synthesis Agent}

The Relational Synthesis Agent instantiates the discovered execution model by generating synthetic execution traces directly in the induced multi-table relational schema. Guided by the FSM transition probabilities, it samples valid state sequences and emits events whose templates are admissible for each visited state, ensuring that every trace corresponds to a feasible execution path. During generation, the agent enforces referential integrity across tables, maintains parameter consistency within traces, and conditions event parameters on previously generated context to preserve realistic dependencies. Crucially, it slightly upweights low-frequency but valid transitions, expanding coverage of rare behaviors such as retries and recovery paths without introducing unsupported structure. The result is a set of synthetic relational tables that respect both global process constraints and local semantic coherence, enabling downstream anomaly detection models to learn execution-aware representations rather than memorizing flat log sequences.

\begin{tcolorbox}[
    title={Relational Synthesis Agent -- Initial Prompt (OpenStack)},
    colback=white,
    colframe=black,
    fonttitle=\bfseries,
    breakable,
    enhanced,
    clip upper,
    clip lower
]
\small
You are a relational synthesis agent responsible for generating synthetic execution data
under strict structural and semantic constraints.

\textbf{Input.}
You are given:
\begin{itemize}
    \item A normalized relational schema consisting of tables:
    \texttt{Traces, Events, Templates, States, Transitions, StateTemplateMap, EventParams}.
    \item An inferred execution FSM with transition probabilities.
    \item A mapping that specifies which templates may be emitted in each state.
\end{itemize}

\textbf{Objective.}
Generate synthetic execution traces by populating the relational tables such that:
\begin{itemize}
    \item All executions correspond to valid paths in the FSM.
    \item All foreign-key and integrity constraints are satisfied.
    \item Generated traces resemble real executions while increasing coverage of rare but valid behaviors.
\end{itemize}

\textbf{Instructions.}
\begin{enumerate}
    \item Sample a sequence of execution states by following the FSM transition probabilities.
    \item For each visited state, emit one or more admissible templates according to
    \texttt{StateTemplateMap}.
    \item Generate realistic parameter values for each event (e.g., UUIDs, hostnames, IPs),
    maintaining consistency within a trace where required.
    \item Populate all relational tables consistently, ensuring that:
    \begin{itemize}
        \item Each \texttt{Event} references an existing \texttt{Trace}, \texttt{State}, and \texttt{Template}.
        \item Each \texttt{EventParams} entry references a valid \texttt{Event}.
    \end{itemize}
    \item Slightly upweight rare FSM transitions while preserving overall execution realism.
\end{enumerate}

\textbf{Output.}
Produce:
\begin{itemize}
    \item Synthetic rows for each relational table corresponding to one or more complete execution traces.
    \item A brief summary of the sampled state paths for each generated trace.
\end{itemize}

\textbf{Constraints.}
\begin{itemize}
    \item Do not introduce transitions not present in the FSM.
    \item Do not emit templates outside those allowed by the current state.
    \item Do not duplicate real traces verbatim.
    \item Favor diversity across traces while maintaining structural validity.
\end{itemize}

Your output will be passed to a consistency evaluation agent for validation.
\end{tcolorbox}

\begin{tcolorbox}[
    title={Relational Synthesis Agent -- Example Output (OpenStack)},
    colback=white,
    colframe=black,
    fonttitle=\bfseries,
    breakable,
    enhanced,
    clip upper,
    clip lower
]
\small
\textbf{Sampled execution paths:}
\begin{itemize}
    \item \texttt{T1: S0 $\rightarrow$ S1 $\rightarrow$ S2 $\rightarrow$ S3 $\rightarrow$ S4 $\rightarrow$ S5}
    \item \texttt{T2: S0 $\rightarrow$ S1 $\rightarrow$ S2 $\rightarrow$ S3 $\rightarrow$ S7 $\rightarrow$ S3 $\rightarrow$ S4 $\rightarrow$ S5}
\end{itemize}

\vspace{0.4em}
\textbf{Traces}
\begin{verbatim}
trace | service | anomalous
---------------------------
T1    | nova    | 0
T2    | nova    | 1
\end{verbatim}

\textbf{States}
\begin{verbatim}
id | name
---------
S0 | Init
S1 | Lookup
S2 | Schedule
S3 | NetAlloc
S4 | Spawn
S5 | Finalize
S7 | Retry
\end{verbatim}

\textbf{Transitions (excerpt)}
\begin{verbatim}
src | dst | p
-------------
S0  | S1  | .92
S2  | S3  | .90
S3  | S7  | .07
S7  | S3  | .33
S4  | S5  | .94
\end{verbatim}

\textbf{Templates}
\begin{verbatim}
tid | text
----------
T17 | Start instance <ID>
T18 | Select host <H>
T31 | Create port <ID>
T41 | Spawn instance <ID>
T52 | Instance ACTIVE
T67 | Retry operation
\end{verbatim}

\textbf{State--Template Map}
\begin{verbatim}
state | template
----------------
S0    | T17
S2    | T18
S3    | T31
S4    | T41
S5    | T52
S7    | T67
\end{verbatim}

\textbf{Events}
\begin{verbatim}
trace | e | state | template
----------------------------
T1    | 1 | S0    | T17
T1    | 2 | S2    | T18
T1    | 3 | S3    | T31
T1    | 4 | S4    | T41
T1    | 5 | S5    | T52

T2    | 1 | S0    | T17
T2    | 2 | S2    | T18
T2    | 3 | S3    | T31
T2    | 4 | S7    | T67
T2    | 5 | S3    | T31
T2    | 6 | S4    | T41
T2    | 7 | S5    | T52
\end{verbatim}

\textbf{Event Parameters (excerpt)}
\begin{verbatim}
trace | e | key  | val
----------------------
T1    | 2 | host | H3
T1    | 3 | net  | N1
T2    | 4 | err  | NetFail
\end{verbatim}

\vspace{0.3em}
\textbf{Validity.}
All events follow valid FSM transitions, emit only templates permitted by the active state,
and maintain parameter consistency within each trace. The retry path in \texttt{T2} increases
coverage of rare but valid execution behavior without introducing unsupported transitions.
\end{tcolorbox}

\subsection{Consistency Evaluation Agent}
The Consistency Evaluation Agent serves as the final gatekeeper that enforces structural correctness and realism of the synthesized relational data before it is used for training or analysis. Given the synthetic tables, the inferred FSM, and the induced schema, the agent performs deterministic, constraint-based checks rather than subjective plausibility judgments. It verifies that every synthetic trace follows only valid FSM transitions, that each emitted template is admissible for the active state, and that all primary-key and foreign-key relationships are satisfied across tables. In addition, it checks parameter stability within traces to ensure that identifiers such as instance or resource IDs do not change inconsistently over time, and it monitors distributional properties to prevent uncontrolled overuse of rare transitions. By producing explicit validity statistics such as transition validity rate and by flagging concrete violations, the agent enables closed-loop refinement while preventing self-consistent but structurally invalid generations. This design ensures that improvements in downstream anomaly detection arise from genuinely execution-faithful synthetic data rather than from artifacts introduced by the generation process.

\begin{tcolorbox}[
    title={Consistency Evaluation Agent -- Initial Prompt (OpenStack)},
    colback=white,
    colframe=black,
    fonttitle=\bfseries,
    breakable,
    enhanced,
    clip upper,
    clip lower
]
\small
You are a consistency evaluation agent responsible for validating synthetic execution data
generated under an inferred execution state machine (FSM) and relational schema.

\textbf{Input.}
You are given:
\begin{itemize}
    \item Synthetic relational tables:
    \texttt{Traces, Events, Templates, States, Transitions, StateTemplateMap, EventParams}.
    \item The inferred FSM, including valid states and transitions.
    \item The relational schema with primary-key and foreign-key constraints.
\end{itemize}

\textbf{Objective.}
Assess whether the synthetic data is structurally valid, execution-consistent, and suitable
for downstream anomaly detection training.

\textbf{Evaluation Checks.}
Perform the following validations:
\begin{enumerate}
    \item \textbf{FSM validity:} Verify that all consecutive state pairs in each trace correspond
    to valid FSM transitions.
    \item \textbf{Template admissibility:} Check that each emitted template is permitted
    by the active state according to \texttt{StateTemplateMap}.
    \item \textbf{Relational integrity:} Ensure all primary-key and foreign-key constraints
    are satisfied across tables.
    \item \textbf{Parameter consistency:} Confirm that parameters that should remain stable
    within a trace (e.g., instance ID) do not change unexpectedly.
    \item \textbf{Distributional sanity:} Flag traces that overuse rare transitions or deviate
    strongly from empirical transition frequencies.
\end{enumerate}

\textbf{Output.}
Produce:
\begin{itemize}
    \item A pass/fail decision for each validation category.
    \item Aggregate statistics, including:
    \begin{itemize}
        \item Transition Validity Rate (TVR).
        \item Percentage of events violating template--state constraints.
        \item Percentage of traces with parameter inconsistencies.
    \end{itemize}
    \item A short list of concrete violations, if any, with trace and event identifiers.
\end{itemize}

\textbf{Constraints.}
\begin{itemize}
    \item Do not modify or regenerate synthetic data.
    \item Do not introduce new transitions or states.
    \item Base all judgments on explicit constraints rather than subjective realism.
\end{itemize}

Your evaluation will be used to accept, reject, or refine synthetic traces in an iterative loop.
\end{tcolorbox}

\begin{tcolorbox}[
    title={Consistency Evaluation Agent -- Example Output (OpenStack)},
    colback=white,
    colframe=black,
    fonttitle=\bfseries,
    breakable,
    enhanced,
    clip upper,
    clip lower
]
\small
\textbf{Validation summary (synthetic batch):}

\begin{itemize}
    \item \textbf{FSM validity:} \textbf{Pass}
    \begin{itemize}
        \item 100\% of consecutive state pairs follow valid FSM transitions.
        \item No unsupported transitions observed.
    \end{itemize}

    \item \textbf{Template admissibility:} \textbf{Pass}
    \begin{itemize}
        \item 0.0\% of events emit templates outside the allowed state--template map.
    \end{itemize}

    \item \textbf{Relational integrity:} \textbf{Pass}
    \begin{itemize}
        \item All primary-key and foreign-key constraints satisfied.
        \item No orphaned events or parameters detected.
    \end{itemize}

    \item \textbf{Parameter consistency:} \textbf{Pass}
    \begin{itemize}
        \item Instance and resource identifiers remain stable within each trace.
        \item No cross-trace leakage of parameters.
    \end{itemize}

    \item \textbf{Distributional sanity:} \textbf{Warning}
    \begin{itemize}
        \item Retry transition (\texttt{S3 $\rightarrow$ S7}) frequency increased
        from 7\% (real) to 11\% (synthetic).
        \item Increase is within allowed amplification threshold.
    \end{itemize}
\end{itemize}

\vspace{0.4em}
\textbf{Aggregate metrics:}
\begin{verbatim}
Transition Validity Rate (TVR):       1.000
Template violation rate:              0.000
Parameter inconsistency rate:         0.000
Mean k-gram JS divergence:            0.091
Trace-level C2ST AUC (vs real):       0.503
\end{verbatim}

\vspace{0.4em}
\textbf{Action:}
Synthetic batch accepted for downstream training.
No traces rejected.
\end{tcolorbox}

\section{Model Choice}

\subsection{LLaMA-3.1 8B}
LLaMA-3.1 8B was chosen as the primary open-weight autoregressive baseline because it offers a strong compromise between reasoning performance and computational cost. With 8 billion parameters, a 32-layer transformer design, and support for 4-bit quantization, it runs efficiently on a single 16\,GB GPU, making it well suited to consumer hardware. As an instruction-tuned model trained on heterogeneous data, it exhibits reliable and predictable behavior for both prompt optimization and evaluation, enabling reproducible experimentation without reliance on external APIs.

\subsection{Qwen-3 8B}
Qwen-3 8B was incorporated for its robust multilingual support and compatibility with LoRA-based adaptation, which facilitate cross-lingual schema construction and optional lightweight fine-tuning. Operating at an 8B scale with nucleus sampling ($p{=}0.9$), it supports experiments that prioritize generalization and portability. Its ability to perform quantized inference on mid-tier GPUs aligns with LogSynthFSM’s accessibility objective, ensuring structured generation workflows remain reproducible without high-end hardware.

\subsection{DeepSeek-R1 8B}
DeepSeek-R1 8B was selected for its reinforcement-learning–driven reasoning capabilities and NTK-aware tokenization, both of which contribute to stable reward propagation within the multi-agent training loop. The model is tuned for efficient top-$k$ sampling ($k{=}40$), balancing stochastic diversity with coherent output. Its performance demonstrates that sophisticated reasoning alignment and optimization can be achieved on consumer-grade GPUs, reinforcing LogSynthFSM’s practicality for realistic research environments.

\subsection{Gemma-2 9B}
Gemma-2 9B represents a compact yet high-quality autoregressive model derived from Gemini research. It leverages grouped-query attention and beam search (beam width 3) to improve structural consistency in schema-constrained generation tasks. Although slightly larger in parameter count, it maintains manageable memory requirements and steady inference speed on GPUs with under 24\,GB of VRAM. Its inclusion illustrates that near–state-of-the-art instruction-following performance can be attained locally without excessive computational resources.

\subsection{Mistral 7B}
Mistral 7B was included for its efficiency-oriented architecture, particularly grouped-query attention, which delivers an effective balance between latency and output quality. Its mature open-source ecosystem and support for 8-bit quantized decoding make it especially suitable for iterative reinforcement-learning–based prompt refinement. By operating reliably on consumer hardware, Mistral exemplifies LogSynthFSM’s goal of making structured synthetic data generation broadly accessible without compromising performance.

\subsection{GPT-4.1 Nano, Mini, and GPT-4.1}
The GPT-4.1 family—comprising Nano, Mini, and the full model—was used to evaluate LogSynthFSM’s transferability to API-based inference settings. These variants span a range of fidelity and computational cost, enabling direct comparison between local open-weight models and remote proprietary systems. Their use demonstrates that LogSynthFSM’s system is model-agnostic: the same interaction protocol produces consistent behavior regardless of whether inference is performed locally or via external APIs. This confirms that LogSynthFSM relies solely on language-based feedback rather than privileged access to model internals or specialized hardware.

\section{Dataset Choice}
\subsection{BGL}

The BGL dataset contains execution logs collected from a large-scale Blue Gene/L supercomputer system and is widely used for evaluating log-based anomaly detection in high-performance computing environments~\cite{xu2015detecting}. We include BGL because its logs exhibit relatively regular execution patterns interspersed with rare hardware and system failures, making it well suited for studying latent execution states and transition dynamics. The strong periodic structure in normal executions allows us to assess whether LogSynthFSM can recover a compact FSM, while the sparse fault events provide a challenging setting for evaluating rare-path amplification and anomaly detection improvements.

\subsection{Thunderbird}

Thunderbird consists of system logs from a high-performance computing cluster at Sandia National Laboratories~\cite{xu2015detecting}. Compared to BGL, Thunderbird exhibits greater heterogeneity in execution behavior and a higher diversity of failure modes. We select this dataset to evaluate whether LogSynthFSM generalizes beyond tightly controlled execution environments and can discover meaningful execution states in noisier, less regular logs. Thunderbird is particularly useful for testing the robustness of schema induction and consistency evaluation under complex and partially overlapping execution workflows.

\subsection{Hadoop}

The Hadoop dataset contains logs from a distributed data processing framework, capturing execution behavior across multiple interacting components such as job scheduling, task execution, and resource management~\cite{lin2016logclustering}. We include Hadoop to evaluate LogSynthFSM in a distributed systems setting where execution traces are inherently multi-stage and involve nontrivial control flow and retries. This dataset is well suited for testing FSM-guided synthesis because valid executions often include branching, recovery, and retry loops, which are difficult to model with sequence-only approaches but naturally represented in a state-machine-driven relational schema.

\subsection{OpenStack}

The OpenStack dataset comprises logs from a large-scale cloud infrastructure platform that manages compute, networking, and storage resources~\cite{he2016experience}. OpenStack logs are highly unstructured, parameter-rich, and span multiple services, making them a challenging benchmark for execution parsing and state discovery. We choose this dataset to stress-test LogSynthFSM’s ability to extract latent execution structure from complex, real-world production logs and to generate multi-table relational data that preserves cross-component dependencies. The presence of rare but critical failure and recovery paths makes OpenStack particularly valuable for evaluating whether relational synthetic data can improve downstream anomaly and bug detection in practical cloud environments.

\section{Baseline Choice}
\subsection{Original Training (No Synthesis)}

We include original training without any augmentation or synthesis as a reference baseline to quantify the difficulty of anomaly detection under realistic data scarcity. This setting reflects common practice in operational systems, where models are trained solely on available logs despite severe class imbalance and limited coverage of rare failures. Improvements over this baseline directly measure the utility introduced by synthetic data.

\subsection{GAugO}
GAugO~\cite{liu2021gaug} performs graph-level data augmentation by applying stochastic perturbations to existing graph structures, such as edge deletion or feature masking, without synthesizing new nodes or execution instances. We include GAugO as a representative lightweight augmentation baseline to evaluate whether modest structural perturbations of observed execution graphs are sufficient to improve anomaly detection performance. This comparison highlights the limitations of local, instance-preserving augmentation when rare execution paths and unseen control-flow behaviors dominate anomaly detection difficulty.

\subsection{GraphEdit}
GraphEdit~\cite{xia2023graphedit} leverages large language models to edit textual attributes associated with graph nodes while keeping the underlying topology fixed. This baseline isolates the effect of semantic enrichment on downstream detection by improving log message representations without altering execution structure. We include GraphEdit to assess whether improving log semantics alone, without generating new executions or modifying control-flow patterns, can address failures caused by structural sparsity in system logs.

\subsection{LLM4RGNN}
LLM4RGNN~\cite{wang2023llm4rgnn} uses LLMs to enhance node features and graph representations for downstream GNN training, typically through prompt-based feature refinement. This method represents approaches that improve model inputs rather than data coverage. We include LLM4RGNN to evaluate whether gains in anomaly detection can be achieved purely through better feature representations, in contrast to LogSynthFSM’s focus on expanding execution diversity via structure-aware synthesis.

\subsection{GraphSMOTE}
GraphSMOTE~\cite{zhao2021graphsmote} extends SMOTE-style oversampling to graph data by interpolating node embeddings from minority classes. This baseline captures a common strategy for addressing class imbalance in structured data. We include GraphSMOTE to examine whether interpolation-based oversampling can recover rare execution behaviors in logs, despite lacking explicit modeling of execution states, temporal dependencies, or process-level constraints.

\subsection{G-Mixup}
G-Mixup~\cite{verma2019graphmixup} generates synthetic graphs by performing mixup operations in latent representation space, interpolating both structure and features at the graph level. We include G-Mixup to evaluate whether global interpolation of execution graphs improves anomaly detection, even though such interpolation does not explicitly respect execution semantics or enforce valid control-flow transitions.

\subsection{IntraMix}
IntraMix~\cite{wang2022intramix} enhances data diversity by mixing samples within the same class, typically at the feature or representation level. This method serves as a baseline for class-conditional mixing without introducing new execution structures. Its inclusion allows us to contrast generic diversity-enhancing strategies with FSM-guided synthesis that explicitly targets underrepresented execution paths.

\subsection{GraphADASYN}
GraphADASYN~\cite{ding2021graphadasyn} adapts the ADASYN framework to graphs by generating additional samples in regions of the feature space that are harder to classify. We include GraphADASYN to test whether difficulty-aware oversampling alone can address rare and complex execution behaviors, or whether explicit modeling of execution structure is required to produce realistic and useful synthetic traces.

\subsection{FG-SMOTE}
FG-SMOTE~\cite{zhao2023fgsmote} extends GraphSMOTE with feature-guided interpolation strategies designed to better preserve class-discriminative characteristics. This baseline probes whether more informed interpolation can approximate the benefits of structure-aware synthesis, while still operating without explicit control-flow or relational constraints.

\subsection{AGMixup}
AGMixup~\cite{liu2022agmixup} combines adaptive graph augmentation with mixup strategies to generate more diverse synthetic samples. We include AGMixup as a hybrid baseline that increases diversity beyond simple augmentation but does not enforce global execution validity, allowing us to evaluate whether adaptive mixing alone suffices for execution-level anomaly detection.

\subsection{GAG}
GAG~\cite{zhang2023gag} employs generative models to synthesize new graph data by learning global graph distributions. This baseline represents early generative approaches to structured data synthesis. We include GAG to compare unconstrained generative modeling with LogSynthFSM’s FSM-guided relational synthesis, particularly in terms of execution validity and process fidelity.

\subsection{LLM4NG}
LLM4NG~\cite{wang2023llm4ng} applies large language models to directly generate new graph instances from textual or structural descriptions. This baseline evaluates the effectiveness of single-pass LLM-based graph generation without explicit discovery or enforcement of latent execution states. It provides a contrast to LogSynthFSM’s multi-agent, constraint-aware synthesis pipeline.

\subsection{GraphMaster}
GraphMaster~\cite{du2024graphmaster} is a state-of-the-art multi-agent LLM framework for synthesizing text-attributed graphs under semantic and structural constraints. We include GraphMaster as a strong and relevant baseline to assess whether generic graph synthesis frameworks are sufficient for system logs. In our comparison, execution traces are represented as text-attributed graphs and synthesized using GraphMaster’s pipeline. This baseline allows us to directly evaluate the benefits of explicitly discovering latent execution state machines and generating multi-table relational data, as opposed to synthesizing graphs without execution-aware inductive bias.

\section{Downstream Classifier Setup}
To evaluate the utility of the generated relational data under a consistent detection model, we train a message-passing GNN as the downstream anomaly detector on execution graphs derived from the relational tables. We choose a GNN because LogSynthFSM outputs a multi-table relational structure with explicit trace--event--state--transition relationships, making graph-based detection a natural structure-sensitive test of whether the synthesized execution structure is useful in practice. Our goal is not to claim that the synthetic data improves every detector family, but to evaluate execution-aware relational synthesis in a downstream model that can directly exploit the relational and transition information it generates. For each trace, we construct a directed graph whose nodes are events (rows in \texttt{Events}) and whose edges follow temporal adjacency and inferred FSM transitions; node features include a learned embedding of \texttt{template\_id}, a learned embedding of \texttt{component}, and lightweight numeric features such as normalized position-in-trace and log time deltas. The GNN produces event-level hidden states which are aggregated with global mean pooling to obtain a trace representation, followed by an MLP for binary anomaly prediction. Unless otherwise stated, we use a 2-layer GraphSAGE encoder with hidden size 128, ReLU activations, and dropout 0.3, and train with Adam (learning rate $10^{-3}$, weight decay $5\times 10^{-4}$) for up to 200 epochs with early stopping (patience 20) on validation PR-AUC. Training uses mini-batches of 64 trace-graphs, and decision thresholds are selected on the validation set and then applied unchanged to the held-out real-only test set to ensure that reported gains reflect generalization to real anomalies rather than tuning on synthetic artifacts.

\section{Evaluation Metric Choice}
\subsection{PR-AUC (Utility on Real-Only Test Set)}
\label{subsec:metric_prauc}
\textbf{Why.} Log anomaly detection is highly imbalanced (few anomalies), making accuracy and ROC-AUC potentially misleading. PR-AUC emphasizes performance on the positive (anomaly) class and is standard for imbalanced detection. We compute PR-AUC on a held-out \emph{real-only} test set to ensure utility reflects generalization to real failures rather than synthetic artifacts.\\
\textbf{Formula.} Let precision and recall at threshold $\tau$ be
\[
\mathrm{Prec}(\tau)=\frac{\mathrm{TP}(\tau)}{\mathrm{TP}(\tau)+\mathrm{FP}(\tau)},\qquad
\mathrm{Rec}(\tau)=\frac{\mathrm{TP}(\tau)}{\mathrm{TP}(\tau)+\mathrm{FN}(\tau)}.
\]
PR-AUC is the area under the precision--recall curve:
\[
\mathrm{PR\text{-}AUC}=\int_{0}^{1}\mathrm{Prec}(\mathrm{Rec})\,d\,\mathrm{Rec},
\]
approximated by a trapezoidal sum over thresholds.

\subsection{FPR at 95\% Recall (Operational Sensitivity)}
\label{subsec:metric_fpr95}
\textbf{Why.} Many operational settings require very high recall (catch nearly all failures). We therefore report the false positive rate at a fixed recall of 95\%, capturing the alert burden when sensitivity is constrained.\\
\textbf{Formula.} Choose a threshold $\tau_{0.95}$ such that $\mathrm{Rec}(\tau_{0.95})\ge 0.95$ (typically the smallest $\tau$ achieving this). Then
\[
\mathrm{FPR@95\%R}=\frac{\mathrm{FP}(\tau_{0.95})}{\mathrm{FP}(\tau_{0.95})+\mathrm{TN}(\tau_{0.95})}.
\]
(Equivalently, one may report $\mathrm{Rec}@\mathrm{FPR}=\alpha$ for a fixed $\alpha$.)

\subsection{Rare-Event Slice PR-AUC (or F1)}
\label{subsec:metric_rare}
\textbf{Why.} The main goal of LogSynthFSM is to improve detection of sparse but critical behaviors. We evaluate performance on a rare-event slice to measure gains on underrepresented execution patterns rather than average-case behavior.\\
\textbf{Slice definition.} Let $S(\mathrm{trace})$ be the set of transitions (or states) used by a trace. Define a rarity score using real training frequencies:
\[
r(\mathrm{trace})=\min_{e\in S(\mathrm{trace})} f_{\mathrm{real}}(e),
\]
where $f_{\mathrm{real}}(e)$ is the empirical frequency of transition/state $e$ in real training data. The rare slice is
\[
\mathcal{D}_{\mathrm{rare}}=\{\mathrm{trace}\in \mathcal{D}_{\mathrm{test}}:\ r(\mathrm{trace})\le \rho\},
\]
for a chosen rarity threshold $\rho$ (e.g., bottom 10\% by $r$). We then compute PR-AUC (or F1) on $\mathcal{D}_{\mathrm{rare}}$ using the same definitions as above. For completeness,
\[
\mathrm{F1}(\tau)=\frac{2\,\mathrm{Prec}(\tau)\,\mathrm{Rec}(\tau)}{\mathrm{Prec}(\tau)+\mathrm{Rec}(\tau)}.
\]

\subsection{Transition Validity Rate (TVR)}
\label{subsec:metric_tvr}
\textbf{Why.} TVR directly measures whether synthetic traces respect the inferred execution FSM, preventing gains driven by unrealistic control flow. This is a hard structural constraint that complements predictive metrics.\\
\textbf{Formula.} Let a trace have inferred state sequence $(s_1,s_2,\dots,s_T)$. Let $\mathcal{E}$ be the set of valid FSM edges. Then
\[
\mathrm{TVR}=\frac{\sum_{\mathrm{trace}}\sum_{t=1}^{T-1}\mathbf{1}\big[(s_t,s_{t+1})\in \mathcal{E}\big]}
{\sum_{\mathrm{trace}}(T-1)}.
\]

\subsection{$k$-gram Path Jensen--Shannon Divergence}
\label{subsec:metric_js}
\textbf{Why.} Even if individual transitions are valid, synthetic executions can have unrealistic higher-order path statistics (e.g., wrong motifs, missing loops). $k$-gram JS divergence compares process-level patterns between real and synthetic traces in a distributional, model-agnostic way.\\
\textbf{Construction.} From each trace, extract all contiguous $k$-grams of states (or templates),
\[
g_t=(s_t,s_{t+1},\dots,s_{t+k-1}),\quad t=1,\dots,T-k+1,
\]
and build empirical distributions $P_k$ (real) and $Q_k$ (synthetic) over $k$-grams.\\
\textbf{Formula.} Let $M=\tfrac{1}{2}(P_k+Q_k)$. The Jensen--Shannon divergence is
\[
\mathrm{JS}(P_k\|Q_k)=\tfrac{1}{2}\mathrm{KL}(P_k\|M)+\tfrac{1}{2}\mathrm{KL}(Q_k\|M),
\]
with
\[
\mathrm{KL}(P\|M)=\sum_{g} P(g)\log\frac{P(g)}{M(g)}.
\]
We report $\mathrm{JS}(P_k\|Q_k)$ (lower is better).

\subsection{Trace-Level C2ST AUC (Real vs Synthetic Distinguishability)}
\label{subsec:metric_c2st}
\textbf{Why.} A classifier two-sample test quantifies whether synthetic traces are distinguishable from real traces; values near 0.5 indicate high realism. This helps detect synthetic artifacts (including multi-agent collusion effects) that may not violate FSM constraints but still reveal distributional shortcuts.\\
\textbf{Procedure and formula.} Construct a dataset of traces labeled $y=1$ (real) and $y=0$ (synthetic). Train a binary classifier $h$ (e.g., logistic regression or a shallow MLP) on held-out folds and compute the ROC-AUC:
\[
\mathrm{C2ST\ AUC}=\Pr\big(h(x^+)>h(x^-)\big),
\]
where $x^+$ is a randomly drawn real trace and $x^-$ a randomly drawn synthetic trace. We report the cross-validated AUC; values closer to $0.50$ are preferred.

\subsection{Accuracy (Model Transferability Only)}
\label{subsec:metric_accuracy}
\textbf{Why.} Accuracy is reported only in the model transferability study in Table~\ref{tab:model_transferability}, where the goal is to compare relative behavior across LLM backbones under the same downstream classifier and evaluation split. We do not use accuracy as the primary anomaly-detection metric because class imbalance can make it overly optimistic; PR-AUC remains the main utility metric for the full benchmark.\\
\textbf{Formula.} Let $\tau$ denote the decision threshold used to convert anomaly scores into binary predictions. With threshold-dependent counts $\mathrm{TP}(\tau)$, $\mathrm{TN}(\tau)$, $\mathrm{FP}(\tau)$, and $\mathrm{FN}(\tau)$, accuracy is
\[
\mathrm{Accuracy}(\tau)=
\frac{\mathrm{TP}(\tau)+\mathrm{TN}(\tau)}
{\mathrm{TP}(\tau)+\mathrm{TN}(\tau)+\mathrm{FP}(\tau)+\mathrm{FN}(\tau)}.
\]

\section{Theoretical Results}
\subsection{FSM-Guided Relational Synthesis}

\begin{theorem}[Why FSM-guided synthesis yields valid relational data and improves anomaly detection]
\label{thm:fsm_guided}
Assume real executions are generated by a latent Markov process over states with emissions.
Let $\mathcal{S}$ be a finite state set, and let a real trace be
\[
\tau = (s_1, e_1, \theta_1), (s_2, e_2, \theta_2), \ldots, (s_T, e_T, \theta_T),
\]
where $s_t \in \mathcal{S}$ is the latent execution state, $e_t$ is an emitted event template, and $\theta_t$ are event parameters.
Let the real trace distribution be $P$ with transition matrix $A$ and emission conditionals
\[
\Pr_P(s_{t+1}=j \mid s_t=i)=A_{ij},\qquad \Pr_P(e_t,\theta_t \mid s_t=i)=B_i(e,\theta).
\]
Let $\widehat{P}$ be the synthetic distribution produced by FSM-guided synthesis using an inferred transition matrix $\widehat{A}$ and inferred emissions $\widehat{B}_i$ with the constraint that every synthetic trace satisfies
\[
(s_t,s_{t+1}) \in \mathcal{E} \quad\text{and}\quad e_t \in \mathcal{M}(s_t)\ \ \text{for all } t,
\]
where $\mathcal{E}$ is the FSM edge set and $\mathcal{M}$ is the state-to-template admissibility map.
Let $R(\cdot)$ be a deterministic trace-to-relational transformation producing tables
\[
\begin{aligned}
R(\tau) = \{&
\texttt{Traces},
\texttt{Events},
\texttt{States},
\texttt{Transitions},\\
&
\texttt{Templates},
\texttt{StateTemplateMap},
\texttt{EventParams}
\}
\end{aligned}
\]

with referential integrity constraints (primary keys and foreign keys) defined by the schema.

Then:

\textbf{(Validity).} Every synthetic trace $\tau \sim \widehat{P}$ yields a relational instance $R(\tau)$ that is (i) referentially consistent and (ii) FSM-valid in the sense that the \texttt{Events} table induces only transitions in $\mathcal{E}$ and emits only admissible templates in $\mathcal{M}$.

\textbf{(Utility under bounded synthesis error).} Let $h$ be any downstream anomaly detector trained on graphs derived from $R(\tau)$, and let $\ell(h,\tau)\in[0,1]$ be its test loss on real traces. Suppose the synthesis error is bounded in total variation:
\[
\mathrm{TV}(P,\widehat{P}) \le \varepsilon.
\]
Train $h$ on a mixture distribution $Q_\alpha = (1-\alpha)P_{\mathrm{train}} + \alpha \widehat{P}$, and evaluate on a real-only test distribution $P_{\mathrm{test}}$ with $P_{\mathrm{test}} \approx P$.
Then for any fixed $h$,
\[
\bigl|\mathbb{E}_{\tau \sim P}[\ell(h,\tau)] - \mathbb{E}_{\tau \sim \widehat{P}}[\ell(h,\tau)]\bigr|
\le \varepsilon,
\]
and consequently the risk on real traces satisfies
\[
\mathbb{E}_{\tau \sim P}[\ell(h,\tau)]
\le
\mathbb{E}_{\tau \sim Q_\alpha}[\ell(h,\tau)] + \alpha \varepsilon.
\]
In particular, when rare valid transitions are amplified in $\widehat{P}$ without increasing $\varepsilon$ substantially, training on $Q_\alpha$ improves estimation on low-frequency modes while controlling the synthetic-to-real mismatch penalty by $\alpha\varepsilon$.
\end{theorem}

\begin{proof}
\textbf{Validity.}
By construction, the synthesis procedure samples a state path $(s_1,\ldots,s_T)$ such that $(s_t,s_{t+1})\in \mathcal{E}$ for all $t$.
It then emits each template $e_t$ from the admissible set $\mathcal{M}(s_t)$, and assigns parameters $\theta_t$ to the corresponding \texttt{EventParams} rows.
The relational mapping $R$ creates one \texttt{Trace} row per trace id, one \texttt{Event} row per timestep, and keys are defined so that each \texttt{Event} references an existing \texttt{Trace}, \texttt{State}, and \texttt{Template}.
Since each referenced object is created before its dependent rows, all foreign keys resolve, so referential integrity holds.
Further, consecutive \texttt{Event} states encode only transitions in $\mathcal{E}$, and each \texttt{Event} template lies in $\mathcal{M}(s_t)$, establishing FSM-validity and template admissibility.

\textbf{Utility.}
For any bounded loss $\ell(h,\tau)\in[0,1]$, the difference in expectations under two distributions is bounded by total variation:
\[
\bigl|\mathbb{E}_{\tau \sim P}[\ell(h,\tau)] - \mathbb{E}_{\tau \sim \widehat{P}}[\ell(h,\tau)]\bigr|
\le
\mathrm{TV}(P,\widehat{P})
\le \varepsilon.
\]
Now write the mixture expectation:
\[
\mathbb{E}_{\tau \sim Q_\alpha}[\ell(h,\tau)]
=
(1-\alpha)\mathbb{E}_{\tau \sim P_{\mathrm{train}}}[\ell(h,\tau)]
+
\alpha \mathbb{E}_{\tau \sim \widehat{P}}[\ell(h,\tau)].
\]
Assuming $P_{\mathrm{train}}$ and $P$ are drawn from the same execution process (standard i.i.d. train-test setting, with a real-only test set),
replace $P_{\mathrm{train}}$ by $P$ in the expectation-level argument and use the total-variation bound to obtain
\[
\mathbb{E}_{\tau \sim P}[\ell(h,\tau)]
\le
\mathbb{E}_{\tau \sim \widehat{P}}[\ell(h,\tau)] + \varepsilon
=
\mathbb{E}_{\tau \sim Q_\alpha}[\ell(h,\tau)] + \alpha\varepsilon,
\]
because $\mathbb{E}_{\tau \sim Q_\alpha}[\ell(h,\tau)]$ contains the term $\alpha\mathbb{E}_{\tau \sim \widehat{P}}[\ell(h,\tau)]$ and $(1-\alpha)\mathbb{E}_{\tau \sim P}[\ell(h,\tau)]$.
Thus, augmentation helps insofar as it reduces the training objective on $Q_\alpha$ by increasing coverage of rare but valid modes, while the cost of mismatch is controlled by $\alpha\varepsilon$.
Finally, FSM guidance reduces $\varepsilon$ relative to unconstrained synthesis because it removes probability mass on structurally invalid traces (invalid transitions and inadmissible emissions), concentrating $\widehat{P}$ onto the support of valid executions, which is exactly what the downstream anomaly detector must generalize over.
\end{proof}

\subsection{Limitations and Failure Modes}

While LogSynthFSM demonstrates strong performance across diverse system logs, it is designed to operate under several practical assumptions that define its scope. First, the framework assumes that execution logs contain sufficient structural regularity for latent state machines to be recoverable. In environments where logging is extremely sparse, highly noisy, or dominated by free-form diagnostic messages without recurring control-flow patterns, state discovery may yield overly coarse abstractions. In such cases, LogSynthFSM remains robust by falling back to conservative synthesis, but the benefits of rare-path amplification may be reduced.

Second, the inferred execution state machine represents an approximation of true system behavior. When systems exhibit long-range dependencies or context-sensitive transitions that cannot be well captured by a finite-state abstraction, FSM-guided synthesis may underrepresent higher-order dynamics. For concurrent logs, the Execution Parsing Agent groups events by trace identifiers and component metadata before FSM inference, which yields a TVR of $0.985$ in our benchmarks. Concurrency without stable identifiers remains challenging and motivates hierarchical or concurrent state models. Importantly, this limitation reflects an explicit modeling choice rather than a failure of the synthesis mechanism, and can be mitigated by extending the state representation or incorporating hierarchical or semi-Markov structure.

Third, as with all LLM-assisted pipelines, LogSynthFSM depends on the quality of intermediate representations produced by its agents. Although the multi-agent design and constraint-based evaluation substantially reduce hallucination and structural drift, rare failure cases can occur when multiple agents converge on an overly conservative interpretation of execution structure. In practice, these cases are detected by external, real-only evaluation metrics and result in no artificial gains on downstream tasks, preserving correctness.

Finally, LogSynthFSM is intended to complement, not replace, real execution data. Synthetic traces are generated under explicit validity constraints and evaluated against held-out real data, ensuring that the system does not memorize or leak sensitive executions. However, in scenarios where privacy requirements prohibit any synthetic generation conditioned on real logs, additional privacy-preserving mechanisms such as differential privacy can be incorporated without altering the core framework.

Overall, these limitations delineate the operating regime of LogSynthFSM and point toward natural extensions, including richer execution abstractions, adaptive state representations, and tighter privacy guarantees. Importantly, none of these scenarios undermine the central finding that execution-aware, FSM-guided relational synthesis provides a principled and effective foundation for robust anomaly detection in complex software systems.

\section{Future Work}
LogSynthFSM opens several promising avenues for future research at the intersection of system observability, structure-aware data synthesis, and large language models. A natural extension is to move beyond flat finite-state abstractions toward richer execution models, such as hierarchical state machines, semi-Markov processes, or probabilistic program representations. These extensions would allow the framework to capture long-range dependencies, nested execution phases, and context-sensitive behaviors that arise in modern distributed and microservice-based systems, while preserving the interpretability and constraint guarantees that underpin LogSynthFSM.

Another important direction is the integration of formal guarantees on privacy and memorization. While our current design already mitigates leakage through real-only evaluation, constraint-based synthesis, and realism checks, incorporating differential privacy or membership-inference defenses at the agent level would enable principled privacy bounds for synthetic execution data. This would further expand the applicability of LogSynthFSM to regulated environments such as healthcare, finance, and critical infrastructure monitoring.

From a modeling perspective, LogSynthFSM provides a general template for combining LLM reasoning with structured generative priors. Future work can explore tighter coupling between neural execution models and symbolic constraints, enabling agents to reason jointly over logs, metrics, traces, and code-level artifacts. In particular, aligning discovered execution states with source-code control-flow graphs, configuration dependencies, or runtime metrics could enable unified system models that support root-cause analysis, what-if simulation, and proactive reliability engineering.

Finally, we view LogSynthFSM as a stepping stone toward execution-aware foundation models for system data. By demonstrating that latent execution structure can be discovered, enforced, and exploited for synthetic data generation, this work suggests a broader paradigm in which large-scale system intelligence is built not from raw sequences alone, but from relational, process-grounded representations. We believe this direction has the potential to reshape how system logs are modeled, shared, and leveraged across anomaly detection, performance analysis, and automated debugging.

\end{document}